\documentclass[10pt,journal,compsoc]{IEEEtran}
\usepackage{graphicx}
\usepackage{amsmath}
\usepackage{amsthm}
\usepackage{amssymb}
\usepackage{booktabs}
\usepackage{algorithm}
\usepackage{algorithmic}
\usepackage{multirow,amsfonts,makecell,array}
\usepackage{threeparttable}
\usepackage{xcolor}
\definecolor{mydarkred}{rgb}{0.6,0,0}
\definecolor{mydarkgreen}{rgb}{0,0.6,0}
\usepackage[colorlinks,
linkcolor=mydarkred,
citecolor=mydarkgreen]{hyperref}

\newtheorem{theorem}{Theorem}
\newcommand{\bm}[1]{\mbox{\boldmath{$#1$}}}

\newcommand{\mb}[1]{\mathbb{#1}}
\newcommand{\myPara}[1]{\vspace{.05in}\noindent\textbf{#1}}
\newcommand{\best}[1]{{\textbf{#1}}}
\newcommand{\ie}{\textit{i}.\textit{e}.}
\newcommand{\eg}{\textit{e}.\textit{g}.}

\graphicspath{{./fig/}}
\usepackage{float}
\usepackage{subfig}
\usepackage{wrapfig}
\usepackage{multirow}
\usepackage{multicol}

\newtheorem{Definition}{Definition}
\newtheorem{Theorem}{Theorem}
\newtheorem{Lemma}{Lemma}

\newtheorem{corollary}{Corollary}

\ifCLASSOPTIONcompsoc
\usepackage[nocompress]{cite}
\else
\usepackage{cite}
\fi
\ifCLASSINFOpdf
\else
\fi
\begin{document}
	%
	\title{Transferring Annotator- and Instance-dependent Transition Matrix for Learning from Crowds}
	%
	%
	%
	%
	
	\author{Shikun~Li,
		Xiaobo~Xia, 
		Jiankang~Deng,~\IEEEmembership{Member,~IEEE,}
		Shiming Ge$^\dagger$, ~\IEEEmembership{Senior Member,~IEEE,} \\
		Tongliang Liu, ~\IEEEmembership{Senior Member,~IEEE}
		
		\thanks{$^\dagger$\quad \ Corresponding author.}
		\IEEEcompsocitemizethanks{
			\IEEEcompsocthanksitem S. Li and S. Ge are with the Institute of Information Engineering, Chinese
			Academy of Sciences, Beijing 100095, China, and also with the School of Cyber
			Security at University of Chinese Academy of Sciences, Beijing 100049,
			China (e-mail: \{lishikun,geshiming\}@iie.ac.cn).
			\IEEEcompsocthanksitem X. Xia and T. Liu are with the Sydney AI Center, School of Computer Science, Faculty of Engineering, The University of Sydney, Darlington, NSW2008, Australia (e-mail: xxia5420@uni.sydney.edu.au;  tongliang.liu@sydney.edu.au).
			\IEEEcompsocthanksitem J. Deng is with the Department of Computing, Imperial College London, London SW7 2BX, United Kingdom (e-mail:
			j.deng16@imperial.ac.uk). }
	}

	\IEEEtitleabstractindextext{%
		\begin{abstract}
			Learning from crowds describes that the annotations of training data are obtained
			with crowd-sourcing services. Multiple annotators each complete their own small
			part of the annotations, where labeling mistakes that depend on annotators occur
			frequently. Modeling the label-noise generation process by the noise transition matrix is a powerful tool to tackle the label noise. In real-world crowd-sourcing scenarios,
			noise transition matrices are both annotator- and instance-dependent. However, due
			to the high complexity of annotator- and instance-dependent transition matrices
			(AIDTM), \textit{annotation sparsity}, which means each annotator only labels a tiny
			part of instances, makes modeling AIDTM very challenging. Without prior knowledge, existing works simplify
			the problem by assuming the transition matrix is instance-independent or using simple parametric ways, which lose modeling generality. Motivated by this, we
			target a more realistic problem, estimating general AIDTM in practice. Without
			losing modeling generality, we parameterize AIDTM with deep neural networks.
			To alleviate the modeling challenge, we suppose every annotator shares its noise
			pattern with similar annotators, and estimate AIDTM via \textit{knowledge transfer}. We
			hence first model the mixture of noise patterns by all annotators, and then transfer
			this modeling to individual annotators. Furthermore, considering that the transfer
			from the mixture of noise patterns to individuals may cause two annotators with
			highly different noise generations to perturb each other, we employ the knowledge transfer between identified neighboring annotators to calibrate the modeling. Theoretical analyses are derived to demonstrate that both the knowledge transfer from global to individuals and the knowledge transfer between neighboring individuals can effectively help mitigate the challenge of modeling general AIDTM.
			Experiments confirm the superiority of the proposed approach on synthetic and
			real-world crowd-sourcing data. The
			implementation is available at \url{https://github.com/tmllab/TAIDTM}.
		\end{abstract}
		
		}

	\maketitle

	\IEEEdisplaynontitleabstractindextext

	%
	\IEEEpeerreviewmaketitle

	\section{Introduction}\label{sec:introduction}
	\IEEEPARstart{I}{n} the era of deep learning, datasets are becoming larger and larger. Nowadays, in multiple fields, large-scale datasets with high-quality annotations are almost the norm to obtain state-of-the-art deep
	learning models~\cite{HanYPTXYS19,song2023tnnls,ServajeanJSCP17,han2022survey,hospedales2021meta,croitoru2023diffusion}. However, it is much expensive to obtain high-quality annotations~\cite{Han2018NIPS,xia2021sample,LiXGL22,li2022estimating,zhou2023asymmetric,fatras2021wasserstein}.
	Thereby, crowd-sourcing is exploited to build large annotated datasets in practice~\cite{Xu2019DeepRS,Li2017ReliableCA,Zhang2020DisentanglingHE,xia2021robust,wei2022aggregate}. The
	crowd-sourcing~\cite{ShengPI08,Zhang22} works obeying the following scenarios: (1) the task issuer divides the
	overall data-annotation task into several sub-tasks that could be overlapped, and distributes them
	to annotators; (2) each annotator then only labels a small fraction of the data; (3) annotations from
	different annotators are collected to achieve final annotated data, where an instance may have multiple
	annotations from different annotators. As many annotators do not have expert knowledge, annotation
	mistakes unavoidably occur, resulting in noisily labeled data~\cite{SnowOJN08,ZhengLLSC17,Sharmanska2016AmbiguityHC}. The research of learning
	from crowds came into being, which aims to learn a classifier robustly from the noisily labeled data.
	
	Existing algorithms for learning from crowds can be generally divided into two categories: model-free
	and model-based algorithms. In the first category, many heuristics mitigate label noise in crowds
	without modeling its generation~\cite{IpeirotisPSW14,Zeng022,SchoenebeckT21,Jiang2022LearningFC}, \eg, identifying clean labels from noisy ones
	via majority voting~\cite{IpeirotisPSW14}. Although these methods empirically work well, their reliability cannot
	be guaranteed without modeling the label noise explicitly. This naturally motivates researchers to
	model and learn label noise by many model-based algorithms. Among these algorithms, the noise
	transition matrix (also called the confusion matrix~\cite{Dawid1979Maximum}), is the most common way to explicitly model
	the generation process of label noise. In real-world scenarios, the noise transition matrices are both
	annotator- and instance-dependent. For example, people may recognize objects according to their
	different familiarity with various characteristics of objects~\cite{xia2020part}.
	
	However, estimating annotator- and instance-dependent transition matrices (AIDTM) has a very high
	complexity, making it much difficult. Specifically, when modeling AIDTM on a training dataset with $C$ classes, $n$ examples, $R$ annotators, and $r$ noisy labels per example, at least $r\times n \times C \times C$ parameters
	need to be estimated. Besides, the sparse annotations provided by each annotator make estimating
	annotator- and instance-dependent transition matrices more challenging. Therefore, prior works
	avoid the difficulty of estimating AIDTM by simplifying the problem. For example, many existing
	works~\cite{Raykar2010LearningFC,Albarqouni2016TMI,Khetan2018iclr,ibrahim2023deep,Tanno2019LearningFN,Chu0W21,MaO20} assume that the noise transition matrix is instance-independent,
	where for each annotator all instances share the same transition matrix. In addition, without prior knowledge, a few existing
	works~\cite{welinder2010multidimensional,ruvolo2010exploiting,yan2014learning,yan2010modeling,bi2014learning} try to simplify the estimation of AIDTM using some simple parametric ways,
	\eg, logistic regression on instance features~\cite{yan2010modeling}. While, in this era of deep learning, crowd-sourcing datasets become
	complex, high-dimensional, and large-scale, where noise patterns also become complex. Hence, for
	such complicated and sparse-annotated data, it is hard to estimate practical AIDTM by existing
	methods due to their over-simplified models, which lose modeling generality.
	
	In this paper, we target the realistic problem, \ie, estimating general AIDTM in practice. Without losing modeling generality, we parameterize AIDTM with deep neural networks. To alleviate the modeling challenge caused by annotation sparsity, we assume that each annotator shares its noise pattern with similar annotators, and propose to perform knowledge transfer to achieve estimating general AIDTM by deep networks. We hence first model the mixture of noise patterns by all annotators with all noisy data, and then transfer this modeling to the individual annotator. Additionally, note that the knowledge transfer from a mixture of noise patterns to individuals may cause two annotators with highly different noise generations to negatively affect each other. We therefore identify the neighboring annotators of an annotator, and use them to calibrate the previously transferred knowledge.
	
	Technically, with all noisily labeled data, we first model instance-dependent transition matrices for
	all instances by using a global-transition deep network. The global-transition deep network models a
	mixture of label-noise generation processes of different annotators. Afterward, with noisily labeled
	data from each annotator, we fine-tune the global noise-transition network to transfer the global knowledge
	to individuals. The individual noise-transition network for each annotator is achieved accordingly,
	which estimates annotator- and instance-dependent transition matrices. Moreover, a similarity graph
	between annotators is constructed by measuring the parameter difference of individual noise-transition
	networks. For each annotator, the individual noise-transition network is rectified by transferring the
	knowledge of the neighboring annotators based on a \textcolor{black}{graph-convolutional-network (GCN)-based} mapping function. As the neighbors of
	an annotator have similar noise patterns, the transfer can help the annotator calibrate the estimation
	of instance-dependent transition matrices. Through the above procedure, a more precise estimation of annotator- and instance-dependent transition matrices is achieved, further leading to enhanced
	classifier robustness.
	
	Before delving into details, we highlight our main contributions as follows:
	
	\begin{itemize}
		\item  We focus on an important problem of learning from crowds, \ie, estimating general annotator and instance-dependent transition matrices. The significance and challenges of handling the problem are carefully analyzed.
		\item We propose to estimate annotator- and instance-dependent transition matrices using deep neural networks via knowledge transfer. The knowledge about the mixed noise patterns of all annotators is extracted and transferred to individuals. Besides, the knowledge about noise patterns of neighboring annotators can be transferred to one annotator to improve the
		transition matrix estimation.
		\item We provide the theoretical analyses to justify the role of knowledge transfer, which shows that the knowledge transfer from global to individuals addresses the challenge that sparse individual annotations cannot train a high-complexity neural network. Besides, the knowledge transfer between neighboring individuals addresses the issue that the transfer from the mixture of noise patterns to individuals may cause two annotators with highly different noise generations to perturb each other.
		
		\item We conduct extensive experiments to support our claims. Empirical results on synthetic and real-world crowd-sourcing data demonstrate the superiority of our transition matrix estimator. Comprehensive ablation studies and discussions are also provided.
	\end{itemize}
	
	\section{Related Works}
	\subsection{Learning from Crowds}
	\label{LFC}
	According to whether to model the noise generation process in crowds, existing algorithms can be generally divided into two categories: model-free ones and model-based ones. We review them as follows.  
	
	\myPara{Model-free algorithms.} Model-free algorithms do not explicitly model the noise generation process. Instead, they identify the true labels via some aggregation rules or discriminative models.
	Representative methods include but do not limit to majority voting~\cite{IpeirotisPSW14}, weighted majority voting~\cite{Li014,LiLGZFH14}, max-margin majority voting~\cite{TianZ15}, tensor factorization methods~\cite{KargerOS11,MaO20}, and end-to-end aggregation~\cite{CachayBD21,wu2023learning}. Although model-free algorithms can work well, their learning objectives are heuristic, and their performance is not guaranteed.
	
	\myPara{Model-based algorithms.} Model-based algorithms explicitly model the noise generation process of each annotator by using a probabilistic model~\cite{WhitehillRWBM09,ZhouPBM12,BachHRR17,ZhouLPM14,WhitehillRWBM09,SimpsonRPS13,YinHZY17,Chen0YC22,ratner2017snorkel}. As model-based algorithms explicitly model the noise generation process, their performance is more reliable. Among them, the most common probabilistic model is the noise transition matrix, which should be annotator- and instance-dependent in real-world scenarios. However, the sparse annotations provided by each annotator make modeling high-complexity annotator- and instance-dependent transition matrices much challenging. Therefore, prior works avoid the difficulty of estimating AIDTM by simplifying the problem. 
	
	Straightly, massive works~\cite{Dawid1979Maximum,KimG12,Raykar2010JMLR,Albarqouni2016TMI,Khetan2018iclr,Tanno2019LearningFN,Wei2022DeepLF,Chen2020StructuredPE,Rodrigues2018aaai}  simplify the modeling problem by assuming that the noise transition matrix is annotator-dependent but instance-independent. For example, the Dawid-Skene (DS) estimator~\cite{Dawid1979Maximum}, uses instance-independent transition matrices to independently model each annotator’s probability of labeling one class by another class. Then, by considering classifier learning, the parameters of classifier and instance-independent transition matrices are estimated jointly by an EM algorithm~\cite{Raykar2010JMLR, Albarqouni2016TMI} or in an end-to-end fashion~\cite{Rodrigues2018aaai}. Several constraints~\cite{Tanno2019LearningFN,Wei2022DeepLF,Chen2020StructuredPE,ibrahim2023deep}, \eg, trace regularization~\cite{Tanno2019LearningFN}, or geometry regularization~\cite{ibrahim2023deep}, are also added to enhance the identifiability of transition matrices and achieve better estimation.
	
	In addition, without prior knowledge, some works~\cite{welinder2010multidimensional,ruvolo2010exploiting,yan2014learning,yan2010modeling,bi2014learning} estimate AIDTM using simple parametric ways. For example,   ~\cite{welinder2010multidimensional,ruvolo2010exploiting} simplify AIDTM by modeling the instance difficulty and the quality of the annotators separately. ~\cite{yan2014learning,yan2010modeling} employ logistic regression to model AIDTM based on the instance feature. 
	~\cite{bi2014learning} introduces some factors such as dedication to increase the modeling ability of the logistic regression model.  
	Note that since the above early methods are studied in binary classification problems and low-dimensional data, they cannot be directly applied to complex and high-dimensional data. 
	Recently, without considering the annotation sparsity, one work~\cite{gao2022learning} simplifies AIDTM to the combination of the convex hull of some pre-defined permutation matrices.  CoNAL~\cite{Chu0W21} estimates a special type of AIDTM, which assumes one common noise pattern is shared by all annotators and decomposes annotation noise into common noise and individual noise.
	
	
	
	Generally speaking, although these simplified models can avoid the difficulty of estimating general AIDTM, they lead to a greatly limited solution space for AIDTM, which loses modeling generality and cannot fit well with the high-complexity noise patterns. 
	Therefore, the problem of estimating general AIDTM in practice is far from being fully studied, and its modeling challenge caused by annotation sparsity has not been solved.
	
	
	\subsection{Instance-dependent Label-noise Learning}
	\label{label-noise}
	Recently, estimating the instance-dependent transition matrix~\cite{xia2020part,yao2021instance,yang2021estimating,ZhuSL21,Wei2022LearningWN,zhu2021second,xia2022extended,jiang2022information} has become a hot topic in label-noise learning, which models all noisy labels are from one annotator. 
	For example, PTD~\cite{xia2020part} estimates the instance-dependent transition matrix by assuming it can be combined from part-dependent transition matrices.  CausalNL~\cite{yao2021instance} improves the identifiability of the instance-dependent transition matrix by exploring the causal structure.
	BLTM~\cite{yang2021estimating} proposes to map instance features into the Bayes label transition matrix using a deep network and achieves state-of-the-art performance. Note that
	these methods need a large number of labeled data from one annotator to train a powerful neural network for further modeling transition matrix. Hence, the annotations sparsity in the learning-from-crowds makes it impossible to learn general AIDTM for each annotator with these methods.
	
	\section{Methodology}
	\subsection{Preliminaries}
	\label{first}
	We begin by fixing some notations. Let $X \subseteq \mathbb{R}^d$
	d denote the random variable of instances, and $\bar{Y} \in[C]=\{1,2, \ldots, C\}$ indicate the random variable of noisy class labels, where $C$ is the number of classes. In the setting of
	learning from crowds, instances are labeled by multiple annotators. Consider a pool of $R$ annotators indexed by $[R]$. For $i$-th instance $\bm{x}_i$ , its annotators are selected randomly, which are denoted by
	$w_i$. Following~\cite{Khetan2018iclr}, for simplicity, we assume that the size of $w_i$ is the same for each instance
	in the following, which is denoted by $r$. The selected annotator $j$ provides a noisy label $\bar{y}_i^j$
	to the
	instance $\bm{x}_i$
	, where the noisy label is related to the annotator, the instance $\bm{x}_i$ , and its latent clean label
	${y}_i^j$
	. We use $\bar{\boldsymbol{y}}_i^{(r)}$ to refer to $\left\{\bar{y}_i^j\right\}_{j \in w_i}$. Note that our algorithm can also be applied when the size of
	$w_i$ varies across instances, which is justified in Section~\ref{exp}.
	
	In this setting, there are two views regarding the distributions of noisily labeled data: (1) the global
	noisy distribution that considers the instances and corresponding noisy labels $(\bm{x}_i, \bar{y}_i^j)$ $ \in(X, \bar{Y})$ come
	from the same distribution $\bar{D}_G$; (2) individual noisy distributions that consider $(\bm{x}_i, \bar{y}_i^j)$ $ \in(X, \bar{Y}^j)$
	are from different distributions $\bar{D}_I^j$.
	
	As discussed, due to the high complexity of annotator- and instance-dependent transition matrices,
	without losing modeling generality, we parameterize them using deep neural networks. Importantly,
	the transition matrices bring the relationship between the Bayes optimal distribution and noisy
	distribution. The main reason is that studying the transition between Bayes optimal distribution
	and noisy distribution is regarded advantageous to that of studying the transition between clean
	distribution and noisy distribution, as analyzed in~\cite{yang2021estimating}.
	
	Deep neural networks have a powerful capacity, which is the potential to model high-dimensional and complex patterns. However, sparse individual annotations cannot effectively train a high-complexity deep
	network in practice. Hence, assuming noise patterns are shared among similar annotators, we consider
	performing knowledge transfer to alleviate such challenge. Specifically, the knowledge about the mixed
	noise pattern of all annotators is extracted (Section~\ref{global}) and transferred to individuals (Section~\ref{ind}).
	Then, the knowledge about noise patterns of neighboring annotators is transferred to one annotator to
	improve the transition matrix estimation (Section~\ref{neighboring}). Finally, we exploit the estimated AIDTM to
	learn Bayes optimal classifier by a statistically consistent algorithm (Section~\ref{loss_correct}).
	\subsection{Training the Global Noise-transition Network}
	\label{global}
	\myPara{Collecting Bayes optimal labels.} We leverage the noisy data distillation method in~\cite{cheng2019learning} (Theorem 2
	therein) to collect a set of distilled examples $\left(\bm{x}, \bar{\bm{y}}^{(r)}, y^{\star}\right)$ out of the noisy dataset, where $y^{\star}$ is the
	inferred theoretically guaranteed Bayes optimal label. Specifically, we can obtain distilled examples
	by collecting all examples whose noisy class posterior on a certain class is larger than a threshold. Readers who are interested in this can refer to~\cite{cheng2019learning} for more details about the Bayes optimal label collection and theoretical guarantees.
	
	\myPara{Global-transition network training.}  With the collected distilled examples, we can effectively model the global instance-dependent transition matrix, bridging the relationship from the Bayes label distribution to the global noisy distribution, which represents the mixed noise pattern of all annotators.
	Following the method BLTM~\cite{yang2021estimating}, we train a deep network parameterized by $\boldsymbol{\theta}_G$ to estimate the
	instance-dependent transition matrices, which represent the probabilities of Bayes optimal labels flip into noisy labels:
	\begin{equation}
		\hat{T}_{p, q}\left(\boldsymbol{x} ; \boldsymbol{\theta}_G\right)=\mathbb{P}\left(\bar{Y}=q \mid Y^{\star}=p, X=\boldsymbol{x} ; \boldsymbol{\theta}_G\right),
	\end{equation}
	where $Y^{\star}$ denotes the random variable of Bayes optimal labels. The deep network parameterized by
	$\boldsymbol{\theta}_G$ takes $\boldsymbol{x}$ as input and output an estimated transition matrix $\hat{T}_{p, q}\left(\boldsymbol{x} ; \boldsymbol{\theta}_G\right)$. The following
	empirical risk on the inferred noisy distribution and the global noisy labels is minimized to learn the
	parameters of the global network $\boldsymbol{\theta}_G$:
	\begin{equation}
		\label{R1}
		\hat{L}_1\left(\boldsymbol{\theta}_G\right)=\frac{1}{m} \sum_{i=1}^m \frac{1}{r} \sum_{j \in w_i} \ell\left(\bar{\boldsymbol{y}}_i^j, \boldsymbol{y}_i^{\star} \cdot \hat{T}\left(\boldsymbol{x}_i ; \boldsymbol{\theta}_G\right)\right),
	\end{equation}
	where $\ell(.)$ is the cross-entropy loss function, $m$ is the number of distilled examples, and $\bar{\boldsymbol{y}}_i^j \in \mathbb{R}^{1 \times C}$ and
	$\boldsymbol{y}_i^{\star} \in \mathbb{R}^{1 \times C}$ are $\bar{y}_i^j$ and $y_i^{\star}$ in the form of one-hot vectors, respectively. \textcolor{black}{Note that according to the analysis in~\cite{yang2021estimating}, given adequate distilled examples, by minimizing the empirical risk, the noise-transition network will model the transition relationship well, and generalize to the non-distilled examples if they share the same pattern with the distilled examples, which is also consistent with our analyses in Section~\ref{justification1}. Besides, a recent work~\cite{liu2023identifiability} has provided a theoretical perspective to explain the identifiability of label-noise transition matrix modeled by deep neural networks. }
	
	\subsection{Learning Individual Noise-transition Networks}
	\label{ind}
	Since each annotator only labels a small part of
	data, the size of the distilled examples for each
	annotator $\left(\boldsymbol{x}, \bar{y}^j, y^{\star}\right)$ is small, which makes it
	hard to learn the individual noise-transition network directly. To address the modeling issue,
	we assume every annotator shares its noise pattern with similar annotators. For example, the
	features that confuse one annotator are likely
	to cause similar annotators to make mistakes.
	This assumption can also be supported by many
	pieces of psychological and physiological evidence, showing that the perception and recognition ability of humans is based on some shared
	mechanisms, \eg, present parts~\cite{palmer1977hierarchical}, familiarity~\cite{wixted2010role}, and the level of specific abilities~\cite{wilmer2012capturing}.
	Therefore, we believe that based on a similar cognitive process, the individual noise pattern will
	be shared among similar annotators.
	
	As the global noise-transition network models a mixture of noise patterns by different annotators,
	we propose transferring the global modeling to individual modeling. Specifically, we fine-tune the
	last layer of the trained global noise-transition network to transfer the knowledge about the global
	noisy distribution for inferring individual noisy distributions (see Fig.~\ref{transfer1}).
	
	The individual noise-transition network of the annotator $j$ is parameterized by $\boldsymbol{\theta}_j$, which takes the
	instance $\bm{x}$ as the input and outputs the annotator- and instance-dependent transition matrix $\hat{T}^j\left(\boldsymbol{x}; \boldsymbol{\theta}_j\right)$:
	\begin{equation}
		\hat{T}_{p, q}^j\left(\boldsymbol{x} ; \boldsymbol{\theta}_j\right)=\mathbb{P}\left(\bar{Y}^j=q \mid Y^{\star}=p, X=\boldsymbol{x} ; \boldsymbol{\theta}_j\right).
	\end{equation}
	The following empirical risk is minimized to learn the last layer’s parameters of the individual noise
	transition network $\boldsymbol{\theta}_j$:
	\begin{equation}
		\label{R2}
		\hat{L}_2\left(\boldsymbol{\theta}_j\right)=-\frac{1}{m_j} \sum_{i=1}^{m_j} \ell\left(\bar{\boldsymbol{y}}_i^j, \boldsymbol{y}_i^{\star} \cdot \hat{T}^j\left(\boldsymbol{x}_i ; \boldsymbol{\theta}_j\right)\right),
	\end{equation}
	where $m_j$ is the number of distilled examples for the annotator $j$.
	\begin{figure}[ht]
		\centering
		\includegraphics[width=1.0\linewidth]{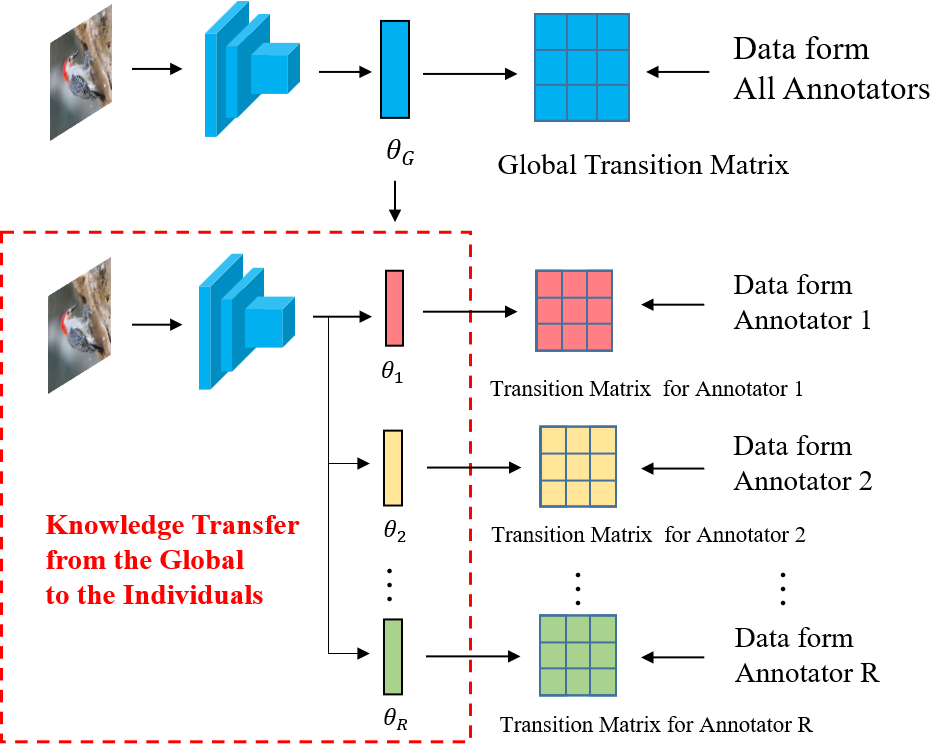}
		\caption{The illustration of transferring the
			global noise-transition network to individual noise-transition network.}
		\label{transfer1}
	\end{figure}
	
	\subsection{Knowledge Transfer between Neighboring Individual Networks}
	\label{neighboring}
			\begin{figure*}
		\centering
		\includegraphics[width=1.0\textwidth]{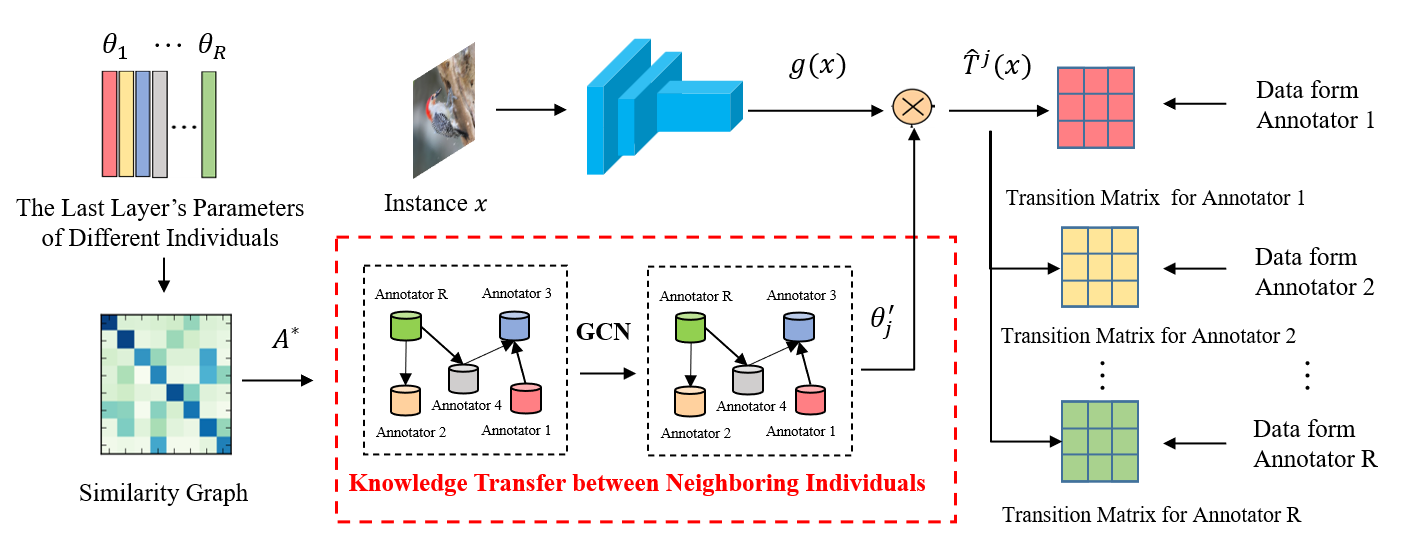}
		\caption{The illustration of the knowledge transfer between neighboring individuals based on a GCN-based mapping function, which merges the node features of neighboring annotators and maps them into the last layer’s parameters of individual noise-transition networks. }
		\label{transfer2}
	\end{figure*}
	 \textcolor{black}{According to the theory of transfer learning~\cite{ben2010theory}, the larger divergence between the source and target domains will lead to a larger error bound in the target domain. It means that when there are annotators with highly different noisy patterns, the large divergence between the mixture of noisy patterns and those individuals may result in a decrease in the effectiveness of knowledge transfer\footnote{\color{black}Note that a more rigorous and specific theoretical justification can be found in Section~\ref{justification1}.}. Motivated by this, we employ the
	knowledge transfer between identified neighboring individuals to improve the estimation.}
	
	\myPara{Construct the similarity graph between annotators.} As the last layer of the individual noise-transition network is a linear layer that maps the latent representations to the noise transition matrix,
	the similarity of the parameters between different individual noise-transition networks can naturally
	represent the similarity between annotators. Inspired by this, we use the cosine distance of the last
	layer’s parameters to measure the similarity:
	\begin{equation}
		\mathbf{S}_{i, j}=\frac{\boldsymbol{\theta}_i \cdot \boldsymbol{\theta}_j}{\left|\boldsymbol{\theta}_i\right|\left|\boldsymbol{\theta}_j\right|},
	\end{equation}
where $i$ and $j$ refer to the annotator $i$ and the annotator $j$, and $|·|$ denotes the L1 norm. Then, we construct the adjacency matrix $\mathbf{A}$ of the similarity graph $G_{\mathbf{S}}$, where nodes are different annotators based on the KNN algorithm~\cite{Cover1967NearestNP}:
\begin{equation}
		\mathbf{A}_{i, j}= \begin{cases}1, & \text { if } j \in \operatorname{NNSearch}(\mathbf{S}, i, k) \\ 0, & \text { otherwise }\end{cases},
		\label{adj}
	\end{equation}
where $\operatorname{NNSearch}(\mathbf{S}, i, k)$ denotes the KNN algorithm that takes the similarity $\mathbf{S}$, the annotator node $i$, and the number of nearest neighbors $k$ as inputs, and output the $k$ nodes nearest to the node $i$.
	
It should be noted that all noise-transition networks are achieved by only employing noisy training data. The construction of the adjacency matrix is unavoidably affected by label errors. To denoise the adjacency matrix, we exploit one of the Graph Purifying methods~\cite{wang2021graph}, \ie, Graph-SVD~\cite{EntezariADP20}. The method Graph-SVD assumes that the adjacency matrix of the true similarity graph is low-ranked, and performs the denoise procedure accordingly. We denote the denoised adjacency matrix as ${\mathbf{A}}^*$, and its
	normalized version as $\widehat{\mathbf{A}}^*$. The visualization results of graph construction can be seen in Section~\ref{vis}.
	
	\myPara{Improve estimations with neighboring annotators.} 
	In order to transfer the knowledge of neighboring annotators to help the estimation of annotator- and instance-dependent transition matrices, inspired by~\cite{ChenWWG19}, we learn inter-dependent parameters for the last layers of individual noise-transition networks $\left\{\boldsymbol{\theta}_j^{\prime}\right\}_{j=1}^R$, via a graph-convolutional-network (GCN)-based mapping function. {\color{black}Intuitively, given an accurate similarity graph, through GCN layers, the node representation of one annotator will be merged with those representations of the nearest similar annotators, which will introduce inter-dependence between them. And thus by employing these node representations to be the last layer’s parameters of individual networks, all data from these nearest similar annotators can be used together to learn their inter-dependent AIDTM. }
	
	Technically, we adopt an $L$-layer GCN to learn the inter-dependent parameters. Let $\mathbf{H}^l \in \mathbb{R}^{R \times Z^l}$  be the feature descriptions of annotator nodes, where $Z^l$ is the dimensionality of each node features in $l$-th layer. According to~\cite{Kipf2017SemiSupervisedCW}, the $l$+1-th layer of GCN takes $\mathbf{H}^l$ and $\widehat{\mathbf{A}}^*$ as inputs, and updates the node features as $\mathbf{H}^{l+1} \in \mathbb{R}^{R \times Z^{l+1}}$, \ie,
	\begin{equation}
		\mathbf{H}^{l+1}=h\left(\widehat{\mathbf{A}}^* \mathbf{H}^l \mathbf{W}^l\right),
	\end{equation}
	where $\mathbf{W}^l \in \mathbb{R}^{Z^l \times Z^{l+1}}$
	is the parameter of $l$+1-th layer to be learned, and $h(\cdot)$ denotes a non-linear operation. Note that the input features of annotator nodes $\mathbf{H}^{0}$ in this work is the concatenation of one-hot vectors representing each annotator. With the help of GCN, the knowledge of neighboring
	annotators is merged and transferred into individuals, and we can obtain final updated node features $\mathbf{H}^L$ as $\left\{\boldsymbol{\theta}_j^{\prime}\right\}_{j=1}^R$.
	By applying the learned inter-dependent weights to the latent instance representations, we can get the improved annotator- and instance-dependent transition matrices as
	\begin{equation}
		\hat{T}^j\left(\boldsymbol{x} ; \boldsymbol{\theta}_j^{\prime}\right)=\boldsymbol{\theta}_j^{\prime} g(\boldsymbol{x}),
		\label{eq8}
	\end{equation}
	where $g(\cdot)$ is the function mapping the instance feature $\bm{x}$ into latent instance representations, which have been
	learned by the global noise-transition network. The following empirical risk is minimized to learn
	GCN’s parameters:
	\begin{equation}
		\label{R3}
		\hat{L}_3\left(\left\{\mathbf{W}^l\right\}_{l=0}^{L-1}\right)=\frac{1}{m} \sum_{i=1}^m \frac{1}{r} \sum_{j \in w_i} \ell\left(\bar{\boldsymbol{y}}_i^j, \boldsymbol{y}_i^{\star} \cdot \hat{T}^j\left(\boldsymbol{x}_i ; \boldsymbol{\theta}_j^{\prime}\right)\right) .
	\end{equation}
	\textcolor{black}{The above technical details and involved notations are summarized in Fig.~\ref{transfer2}. To help to understand, we provide an illustrative example in Appendix~\ref{example}, and the theoretical justification for the role of the GCN-based mapping function can be seen in Section~\ref{theory2}.} Besides, the choice of mapping function
	can be seen in Appendix~\ref{f2}.
	\subsection{Classifier Training with Loss Correction}
	\label{loss_correct}
	After the above learning procedure of individual noise-transition networks via knowledge transfer, we can obtain annotator- and instance-dependent transition matrix $\hat{T}^j\left(\boldsymbol{x}_i; \boldsymbol{\theta}_j^{\prime}\right)$. Following Forward-correction~\cite{PatriniRMNQ17} that is a typical classifier-consistent algorithm in label-noise learning, we minimize the empirical risk as follows to optimize the classification network parameter $\boldsymbol{\phi}$:
	\begin{equation}
		\label{R4}
		\hat{L}_4(\boldsymbol{\phi})=-\frac{1}{n} \sum_{i=1}^n \frac{1}{r} \sum_{j \in w_i} \ell\left(\bar{\boldsymbol{y}}_i^j, f\left(\boldsymbol{x}_i ; \boldsymbol{\phi}\right) \cdot \hat{T}^j\left(\boldsymbol{x}_i ; \boldsymbol{\theta}_j^{\prime}\right)\right),
	\end{equation}
	where $n$ is the number of training examples, and $f (.)$ is a classification network parameterized
	by $\boldsymbol{\phi}$ that aims to predict the Bayes class posterior probability $\mathbb{P}\left(Y^{\star} \mid X\right)$. Besides, inspired by
	T-Revision~\cite{xia2019anchor}, we further tune the noise-transition network together with the classifier. Algorithm~\ref{alg:algorithm1}
	lists the pseudo-code of the proposed method.
	
	\begin{algorithm}[ht]
		\renewcommand{\algorithmicrequire}{\textbf{Input:}}
		\renewcommand{\algorithmicensure}{\textbf{Output:}}
		\caption{Transferring Annotator- and Instance-dependent Transition Matrices for Learning from Crowds.} 
		\label{alg:algorithm1}
		\begin{algorithmic}[1]
			\REQUIRE Noisy dataset $\{\bm{x}_i, \bm{\bar{y}}_{i}^{(r)}\}_{i=1}^n$, the number of nearest neighbors $k$, the number of annotators $R$, and random initialized classifier $f(\cdot ; \phi)$.
			\ENSURE The learned deep classifier $f(\cdot ; \phi)$.
			\STATE {// Training the Global Noise-transition Network} 
			\STATE Collecting distilled examples as ~\cite{cheng2019learning}.
			\STATE Minimize the $\hat{L}_{1}(\bm{\theta}_G)$ in Eq.~(\ref{R1}) on the distilled examples to learn the global noise-transition network.
			\STATE {// Learning Individual Noise-transition Networks} 
			\FOR {$j=1,2,...,R$}
			\STATE Minimize the $\hat{L}_{2}(\bm{\theta}_j)$ in Eq.~(\ref{R2}) on the distilled examples for annotator $j$ to learn the individual noise-transition network for annotator $j$.
			\ENDFOR
			\STATE {// Knowledge Transfer between Neighboring Individual Networks} 
			\STATE Construct the similarity graph between annotators by Eq.~(\ref{adj}).
			\STATE Minimize the $\hat{L}_{3}(\{\mathbf{W}^{l}\}_{l=0}^{L-1})$ in Eq.~(\ref{R3}) to improve the estimations of annotator- and instance-dependent noise transition matrices with neighboring annotators.
			\STATE {// Classifier Training with Loss Correction} 
			\STATE  Minimize the $\hat{L}_{4}(\bm{\phi})$ in Eq.~(\ref{R4}) on noisy training data to learn the classifier.
		\end{algorithmic}
	\end{algorithm}
	
	\section{Theoretical Justification}
	\label{justification}
	In this section, we present the theoretical analyses to justify the role of knowledge transfer in our proposed method. 

	\subsection{Knowledge Transfer from Global to Individuals}
	\label{justification1}
	To demonstrate the importance of the knowledge transfer from global to individuals, we derive generalization error bounds without and with the knowledge transfer.
	
	For convenience, similar to Section~\ref{first}, let the distribution of $(X, \bar{Y}, Y^{\star})$ be ${D}_G$ and the distribution of $(X, \bar{Y}^j, Y^{\star})$ be ${D}_I^j$. Correspondingly, there are two views regarding the collected distilled examples: (1) the global distilled set $\mathcal{D}_G$ with a size $m r$ of examples $(\boldsymbol{x}, \bar{y}, y^{\star})$ from ${D}_G$; (2) $R$ individual distilled sets, where the $j$-th set $\mathcal{D}_I^j$ has a size $m_j$ of examples $(\boldsymbol{x}, \bar{y}^j, y^{\star})$ from ${D}_I^j$.
 
	Let $\mathcal{T}$ be the hypothesis space of the learnable transition matrix $T$ such that $T \in \mathcal{T}$. The cross-entropy loss function $\ell(\bar{\boldsymbol{y}}, \boldsymbol{y}^{\star} \cdot T(\boldsymbol{x}))$ measures the performance of $T$ for a single
	data point $(\boldsymbol{x}, \bar{\boldsymbol{y}}, \boldsymbol{y}^\star)$, and assume it is upper bounded by $M$. The learned global transition matrix and the $j$-th individual transition matrix are $\hat{T}$ and $\hat{T}^j$, respectively. Assume $\mathcal{D}_G$ and $\mathcal{D}_I^j$ are class-balanced \textit{w.r.t.}  $y^{\star}$. 
 
	\myPara{Generalization error bound without knowledge transfer.}
	For theoretical analysis, we define the expected loss of the transition matrix $T$ over all data points that follow the distribution $D$ as $\mathcal{R}_{{D}}(T)=\mathbb{E}_{(\bm{x}, \bar{\bm{y}}, \bm{y}^{\star}) \sim {D}}[\ell(\bar{\boldsymbol{y}}, \boldsymbol{y}^{\star}\cdot T(\boldsymbol{x}))]$, and $\hat{\mathcal{R}}_{\mathcal{D}}(T)$
	be the corresponding empirical loss. Also, we assume the noise-transition network has $d$ layers, parameter matrices ${W}_1,\ldots,{W}_d$, and activation functions $\sigma_1,\ldots,\sigma_{d-1}$ for each layer. The mapping of the noise-transition network is denoted by $t: \boldsymbol{x}\mapsto {W}_d\sigma_{d-1}({W}_{d-1}\sigma_{d-2}(\ldots \sigma_1({W}_1\boldsymbol{x})))\in \mathbb{R}^{C\times C}$. Then, the $(i, j)$-th entry of the transition matrix $T$ is obtained by $T_{ij}(\boldsymbol{x})=\exp \left(t_{ij}(\boldsymbol{x})\right) / \sum_{k=1}^C \exp \left(t_{ik}(\boldsymbol{x})\right)$.

	\begin{theorem} Assume the Frobenius norm of the weight matrices ${W}_1,\ldots,{W}_d$ are at most $M_1,\ldots, M_d$, and the instances $\boldsymbol{x}$ are upper bounded by $B$, \ie, $\|\boldsymbol{x}\|\leq B$ for all $\boldsymbol{x}\in\mathcal{{X}}$. Let the activation functions be 1-Lipschitz, positive-homogeneous, and applied element-wise (such as the ReLU). Then, for any $j \in [R]$, $\delta\in (0,1)$, with probability at least $1-\delta$,
		\label{thm1}
		\begin{align}
			&
			{\mathcal{R}}_{D_I^j}(\hat{T}^j)-\hat{\mathcal{R}}_{\mathcal{D}_I^j}(\hat{T}^j) \nonumber \\ 
			&\leq \frac{2{C}^2B(\sqrt{2d\log2}+1)\Pi_{i=1}^{d}M_i}{\sqrt{Cm_j}}+M\sqrt{\frac{\log({{1}/{\delta}})}{2m_j}}.\nonumber
		\end{align}
	\end{theorem}
	The proof of Theorem~\ref{thm1} is provided in Appendix~\ref{proof_1}. 
	
	\myPara{Remark~1.} Theorem~\ref{thm1} shows that {\color{black} as training a deep classifier network~\cite{mohri2018foundations}, when the size of distilled training examples from one annotator is large, \eg, it annotated all instances, by minimizing the empirical risk, the learned transition matrices will generalize well on unseen test data with the same noisy pattern. However,} for a small number of training examples from the $i$-th annotator, \ie, a small $m_j$, the generalization error may greatly increase due to the high-complexity of a deep network, \ie, large $d$ and $M_i$. This analysis also partly explains why prior works employed highly simplified models~\cite{Raykar2010LearningFC,Albarqouni2016TMI,welinder2010multidimensional,ruvolo2010exploiting,yan2014learning}.
	
	\myPara{Generalization error bound with knowledge transfer.}
	Following~\cite{mcnamara2017risk}, let $\tilde{{T}}_{G}$ be a stochastic hypothesis (\ie, a distribution over $\mathcal{T}$) associated with the learned global transition matrix $\hat{{T}}$, and fine-tuning the global network on the individual data is seen as learning a stochastic hypothesis on $\mathcal{D}_I^j$ with the hypothesis space $\tilde{\mathcal{T}}\subseteq {\mathcal{T}}$ and the prior $\tilde{{T}}_{G}$. Let $\mathcal{R}'_{{D}}(\tilde{T}):=\mathbb{E}_{(\bm{x}, \bar{\bm{y}}, \bm{y}^{\star}) \sim {D}, T\sim \tilde{T}}[\ell(\bar{\boldsymbol{y}}, \boldsymbol{y}^{\star}\cdot T(\boldsymbol{x}))]$,
	and $\hat{\mathcal{R}}'_{\mathcal{D}}(\tilde{T})$
	be the corresponding empirical loss about $\tilde{T}$. 
	\begin{theorem} 
		\label{thm2}
		Suppose given the learned global transition matrix $\hat{T} \in \mathcal{T}$ and $\mathcal{R}_{D_G}(\hat{T})$ estimated from $\mathcal{D}_G$, it is possible to construct $\tilde{\mathcal{T}}$ with the property: $\forall \tilde{T} \in \tilde{\mathcal{T}}, K L\left(\tilde{T} \| \tilde{T}_{G}\right) \leq \omega\left(\mathcal{R}_{D_G}(\hat{T})\right)$, where $KL(\cdot)$ is Kullback-Leibler divergence and $\omega: \mathbb{R} \rightarrow \mathbb{R}$ is a non-decreasing function measuring a transferability property obtained from the knowledge transfer from global to individuals. Assume the empirical loss of $\hat{T}$ is small, \ie, $\hat{\mathcal{R}}_{\mathcal{D}_G}(\hat{T}) \approx 0$. Then, for any $j \in [R]$, $\delta\in (0,1)$, $ \tilde{T}^j \in \tilde{\mathcal{T}}$, with probability at least $1-\delta$, 
		\begin{align}
			{\mathcal{R}}'_{D_I^j}(\tilde{T}^j)-\hat{\mathcal{R}}'_{\mathcal{D}_I^j}(\tilde{T}^j) \nonumber & \leq M\sqrt{\frac{\omega\left({\mathcal{R}}_{D_G}(\hat{T}) \right)+\log({2 m_j}/{\delta})}{2\left(m_j-1\right)}}\nonumber \\ 
			&\leq M\sqrt{\frac{\omega\left(O({1}/\sqrt{mr}) \right)+\log({2 m_j}/{\delta})}{2\left(m_j-1\right)}}.\nonumber
		\end{align}
	\end{theorem}
	The proof of Theorem~\ref{thm2} is provided in Appendix~\ref{proof_2}. 
	
	\myPara{Remark~2.} Theorem~\ref{thm2} shows that when the knowledge transfer from global is effective, \ie, $\omega\left({\mathcal{R}}_{D_G}(\hat{T}) \right) \leq \omega\left(O({1}/\sqrt{mr}) \right)$ is small, the generalization error bound about the $j$-th individual transition matrix will mainly depend on  $M\sqrt{\frac{\log({2 m_j}/{\delta})}{2\left(m_j-1\right)}}$, which is independent of the complexity of the model and enables the modeling by deep neural networks from sparse individual annotations. Moreover, as stated in~\cite{mcnamara2017risk}, $\omega$ quantifies how large $\tilde{\mathcal{T}}$ must be, in terms of $\mathcal{R}_{D_G}(\hat{T})$, so that $\exists \tilde{T}^j \in \tilde{\mathcal{T}}, {\mathcal{R}}'_{D_I^j}(\tilde{T}^j) \leq \epsilon$, where $\epsilon$ is a small constant. {\color{black}It means the effect of this knowledge transfer depends on how close a mixture of noise patterns and individual noise patterns are. In other words,}
	when there are annotators with highly different noisy patterns, the mixture of noisy patterns will be far from certain individuals, leading to more large $\omega\left({\mathcal{R}}_{D_G}(\hat{T}) \right)$, even though ${\mathcal{R}}_{D_G}(\hat{T})$ is small.
	Note that this analysis also justifies the motivation for that we further improve the matrix estimation via the knowledge transfer between neighboring individuals.
	
	\subsection{Knowledge Transfer between Neighboring Individuals}
	\label{theory2}
	In this subsection, to show the important influence of the knowledge transfer between neighboring individuals, we theoretically analyze the roles of GCN-based mapping function, which performs the knowledge transfer.
	
	To simplify the analysis, we assume that each annotator node in the adjacency matrix $\mathbf{A}^*$ has the same number of neighbors $k$, the non-linear function $h(\cdot)$ is the ReLU activation function denoted by $ReLU(\cdot)$, and the spectral norm of the parameter matrix $\mathbf{W}^{l}$ is less than $k$. Then, we have the theorems as follows.
	
	\begin{theorem} The GCN layer can make the node features of similar annotators close.
		\label{the1}
	\end{theorem}
	
	\begin{proof}  Without loss of generality, we focus on the $l$+1-th layer of GCN, which takes the node features of annotators $\mathbf{H}^l$ and the normalized adjacency matrix $\widehat{\mathbf{A}}^*$ as inputs, and updates the node features as $\mathbf{H}^{l+1}$. Let $\mathbf{h}^l_i$ be the $i$-th row of $\mathbf{H}^l$, representing the node feature of annotator $i$ after the $l$-th layer, and $\mathcal{N}_i$ be the neighbors of annotator $i$.
		
		As $\mathbf{H}^{l+1}=h\left(\widehat{\mathbf{A}}^* \mathbf{H}^l \mathbf{W}^l\right)$, we have
		\begin{equation}
			\mathbf{h}^{l+1}_i=h\left(\sum_{j=1}^R \mathbf{A}^*_{ij} \frac{\mathbf{h}^l_j}{k} \mathbf{W}^l\right)=ReLU\left(\sum_{j \in \mathcal{N}_i} \frac{\mathbf{h}^l_j}{k} \mathbf{W}^l\right).
		\end{equation}
		Hence, according to the topology structure of annotator $i$ and annotator $j$, we can divide $\mathbf{h}^{l+1}_i$ into three parts: the node feature $\frac{\mathbf{h}^{l}_i}{k}$ ,  the sum of common neighbor features $\mathcal{S}=\sum_{p \in \mathcal{N}_i \cap \mathcal{N}_j} {\mathbf{h}^l_p}$, and the sum of non-common neighbor features $\mathcal{Q}_i=\sum_{q \in \mathcal{N}_i-\mathcal{N}_i \cap \mathcal{N}_j} {\mathbf{h}_q^l}$.
		
		Then, we have
		\begin{equation}
			\begin{aligned}
				&\quad\left\|{\mathbf{h}}^{l+1}_i-{\mathbf{h}}^{l+1}_j\right\|_2 \\& = \left\|ReLU\left(\sum_{q \in \mathcal{N}_i} \frac{\mathbf{h}^l_q}{k} \mathbf{W}^l\right)-ReLU\left(\sum_{q \in \mathcal{N}_j} \frac{\mathbf{h}^l_q}{k} \mathbf{W}^l\right)\right\|_2\\
				& \leq \left\|\sum_{q \in \mathcal{N}_i} \frac{\mathbf{h}^l_q}{k} \mathbf{W}^l-\sum_{q \in \mathcal{N}_j} \frac{\mathbf{h}^l_q}{k} \mathbf{W}^l\right\|_2 \\
				&\leq \left\|\sum_{q \in \mathcal{N}_i} \frac{\mathbf{h}^l_q}{k}-\sum_{q \in \mathcal{N}_j} \frac{\mathbf{h}^l_q}{k}\right\|_2 \left\|\mathbf{W}^l\right\|_2\\
				& \leq \left\|\left(\frac{\mathbf{h}_i^l}{k}-\frac{\mathbf{h}_j^l}{k}\right)+\left(\frac{\mathcal{S}}{{k}}-\frac{\mathcal{S}}{{k}}\right)+\left(\frac{\mathcal{Q}_i}{{k}}-\frac{\mathcal{Q}_j}{{k}}\right)\right\|_2 \left\|\mathbf{W}^l\right\|_2 \\
				& =\left\|\frac{1}{k}\left({\mathbf{h}_i^l}-{\mathbf{h}_j^l}\right)+\frac{1}{{k}}\left({\mathcal{Q}_i}-{\mathcal{Q}_j}\right)\right\|_2 \left\|\mathbf{W}^l\right\|_2.
			\end{aligned}
		\end{equation}
		Therefore, when annotator $i$ and annotator $j$ are similar, which means that they are neighbors to each other and have similar neighbors (\ie, $\mathcal{N}_i \approx \mathcal{N}_j $ or  ${\mathcal{Q}_i}\approx{\mathcal{Q}_j}$),   then $\left\|{\mathbf{h}}^{l+1}_i-{\mathbf{h}}^{l+1}_j\right\|_2 \approx \left\|\frac{1}{k}\left({\mathbf{h}_i^l}-{\mathbf{h}_j^l}\right)\right\|_2 \left\|\mathbf{W}^l\right\|_2 < \left\|{\mathbf{h}}^{l}_i-{\mathbf{h}}^{l}_j\right\|_2 $.  
	\end{proof}
	\myPara{Remark~3.} This theorem shows that after the node features of annotators pass through the $l$+1-th GCN layer if two annotators have large similarity, GCN will force their features to be close to each other.  By stacking more GCN layers, the higher-order neighbor information will be merged, and the GCN-based mapping function will make the node features of similar annotators close.
	
	\begin{theorem}
		If the input node features of dissimilar annotators are orthogonal, the GCN layer can keep their node features orthogonal.
		\label{the2}
	\end{theorem}
	
	\begin{proof}
		Following the notations in the proof of Theorem 1,  for the $l$+1-th layer of GCN, there are two dissimilar annotators $i$ and $j$,  their input node features are approximately orthogonal, \ie,  $ {\mathbf{h}}^{l}_i\cdot{\mathbf{h}}^{l}_j= 0$. Besides,  if the similarity graph is accurate, the neighbors of annotators $i$ and $j$ are also dissimilar, \ie,  $ {\mathbf{h}}^{l}_q\cdot{\mathbf{h}}^{l}_p= 0, q \in \mathcal{N}_i, p \in \mathcal{N}_j$. Then, we have 
		\begin{equation}
			\begin{aligned}{\mathbf{h}}^{l+1}_i\cdot{\mathbf{h}}^{l+1}_j & = ReLU\left(\sum_{q \in \mathcal{N}_i} \frac{\mathbf{h}^l_q}{k} \mathbf{W}^l\right)\cdot ReLU\left(\sum_{q \in \mathcal{N}_j} \frac{\mathbf{h}^l_q}{k} \mathbf{W}^l\right)\\
				&= \sum_{q \in \mathcal{N}_i} \frac{\mathbf{h}^l_{q}}{k} \mathbf{W}^l_{+} \cdot \sum_{q \in \mathcal{N}_j} \frac{\mathbf{h}^{l}_{q}}{k} \mathbf{W}^l_{+\prime} \\
				&= \left(\sum_{q \in \mathcal{N}_i} \frac{\mathbf{h}^l_{q}}{k} \mathbf{W}^l_+ \right)^\top \sum_{q \in \mathcal{N}_j} \frac{\mathbf{h}^{l}_{q}}{k} \mathbf{W}^l_{+\prime}\\
				&= \frac{1}{k^2} \left(\mathbf{W}^l_+\right)^T \left(\sum_{q \in \mathcal{N}_i} \mathbf{h}^l_{q}\right)^\top \sum_{q \in \mathcal{N}_j}{\mathbf{h}^{l}_{q}} \mathbf{W}^l_{+\prime}\\
				&= \frac{1}{k^2} \left(\mathbf{W}^l_+\right)^\top \mathbf{O} \mathbf{W}^l_{+\prime} =0,
			\end{aligned}
		\end{equation}
		where $\mathbf{W}^l_{+}$ ($\mathbf{W}^l_{+\prime}$) means keeping the values of the columns in $\mathbf{W}^l$ that its corresponding values in  ${\mathbf{h}}^{l+1}_i$ (${\mathbf{h}}^{l+1}_j$) are positive, while setting the values of other columns to zero; $\mathbf{O}$ is the zero matrix.
	\end{proof}
	\myPara{Remark~4.} This theorem shows that given orthogonal one-hot vectors as the input node features of annotators, the GCN-based mapping function will keep the node features of dissimilar annotators away from each other if the feature dimension is sufficiently large.
	\begin{corollary}
		\textcolor{black}{By employing the final node features of annotators to be the last layer’s parameters of individual noise-transition networks}, the estimated AIDTM of similar annotators will be close and those of dissimilar annotators will stay away from each other in usual cases.
	\end{corollary}
	\begin{proof}
		According to Eq.~(\ref{eq8}), $\hat{T}^j\left(\boldsymbol{x} ; \boldsymbol{\theta}_j^{\prime}\right)=\boldsymbol{\theta}_j^{\prime} g(\boldsymbol{x})=\mathbf{h}_j^L g(\boldsymbol{x})$, then $\hat{T}^j\left(\boldsymbol{x} ; \boldsymbol{\theta}_j^{\prime}\right)-\hat{T}^i\left(\boldsymbol{x} ; \boldsymbol{\theta}_i^{\prime}\right)=(\mathbf{h}_j^L-\mathbf{h}_i^L)g(\boldsymbol{x})$. Therefore, according to Theorem~\ref{the1}, when two annotators are similar,  after the GCN-based mapping function, their node features are close, \ie, $\mathbf{h}_j^L - \mathbf{h}_i^L \approx \mathbf{0}$, then $\hat{T}^j\left(\boldsymbol{x}; \boldsymbol{\theta}_j^{\prime}\right) \approx\hat{T}^i\left(\boldsymbol{x}; \boldsymbol{\theta}_i^{\prime}\right)$. Similarly,  according to Theorem~\ref{the2}, when two annotators are dissimilar, after the GCN-based mapping function, their node features stay away from each other, \ie, $\mathbf{h}_j^L$ is highly different from $\mathbf{h}_i^L$, then $\hat{T}^j\left(\boldsymbol{x}; \boldsymbol{\theta}_j^{\prime}\right)$is highly different from $\hat{T}^i\left(\boldsymbol{x}; \boldsymbol{\theta}_i^{\prime}\right)$ in usual cases.
	\end{proof}
	\myPara{Remark~5.} This corollary shows that applying the learned node features as the parameters of the last layers of individual noise-transition networks can introduce inter-dependence into the learned AIDTM,  and then by sufficient training individual noise-transition networks, knowledge transfer between neighboring individuals will help similar annotators to model its AIDTM and keep away from dissimilar annotators, avoiding the perturbation from highly different noisy patterns.
	
	\section{Computational complexity analysis}	
{\color{black} First of all, we assume the number of training examples is $n$,  the average number of annotations for instances is $\bar{r}$, the number of classes is $C$, the number of annotators is $R$, the number of neighbors for each annotator in the similarity graph is $k$, the dimension of the last layer's parameter of the individual noise-transition network is $d\times C\times C$, the number of GCN layers is $L$, the training batch size is $b$, and the computational complexity of one instance forward passes through the classifier network is $O(T)$. For convenience, we further assume $L=2$ and the dimension of the node features of the first GCN layer is $Z_1$.
	
	According to Algorithm~\ref{alg:algorithm1},  our method performs three steps to estimate general AIDTM, and the time complexities of these steps are as follows:
	\begin{enumerate} 
		\item When training a global noise-transition network, since we use the same model structure as the classifier network except for the last layer,  the time complexity is about $O(2nTE_1)$, where $E_1$ is the number of training epochs. 
		\item When learning individual noise-transition networks by finetuning the last layer,  the time complexity is about $O(nTE_2+2n\bar{r}d^2C^2E_2)$, where $E_2$ is the number of finetuning epochs. 
		\item For the knowledge transfer between neighboring individuals, as the time complexity of KNN search~\cite{cover1967nearest} is $O(dC^2R^2+kR^2)$ and the time complexity of GCN is $O(dkC^2R^2Z_1)$,  its time complexity is $O(dC^2R^2+kR^2+ 2ndkC^2R^2Z_1E_3/b+nTE_3+2n\bar{r}d^2C^2E_3)$, where $E_3$ is the number of training epochs for knowledge transfer.  
	\end{enumerate}
    Besides, when learning a deep classifier with loss correction, given estimated AIDTM, the time complexity is about $O(2nTE)$, where  $E$ is the number of classifier training epochs.
    
	Hence,  the overall time complexity of the proposed method is about $O(nT(2E_1+E_2+E_3+2E)+2n\bar{r}d^2C^2(E_2+E_3)+2ndkC^2R^2Z_1E_3/b+(dC^2+k)R^2)$, which is linearly related to the number of training examples $n$.}	
	\begin{table*}[ht]
		\caption{Comparison with state-of-the-art methods in the test accuracy (\%) on simulated FashionMNIST and K-MNIST datasets. The best results are in~\best{bold}.}
		\centering
		\fontsize{7.5pt}{10pt}\selectfont
		\setlength\tabcolsep{2.5pt}
		\begin{tabular}{c|ccccc|c||ccccc|c}
			\hline \multirow{2}{*}{ Dataset } & \multicolumn{6}{c||}{ F-MNIST } & \multicolumn{6}{c}{ K-MNIST } \\
			\cline{2-13} 
			& AIDN-10\% & AIDN-20\% & AIDN-30\% & AIDN-40\% & AIDN-50\% & Avg. & AIDN-10\% & AIDN-20\% & AIDN-30\% & AIDN-40\% & AIDN-50\% & Avg. \\
			\hline CE &  90.70$\pm$0.27  &  88.57$\pm$0.27  &  84.39$\pm$0.36  &  76.73$\pm$0.81  &  68.12$\pm$1.36  & 81.70 &  94.14$\pm$0.20  &  91.24$\pm$0.36  &  86.43$\pm$0.40  &  78.81$\pm$0.30  &  71.38$\pm$0.09  & 84.40 \\
			GCE &  92.80$\pm$0.11  &  92.17$\pm$0.20  &  90.32$\pm$0.37  &  84.31$\pm$2.20  &  74.70$\pm$2.47  & 86.86 &  96.83$\pm$0.24  &  96.57$\pm$0.08  &  95.33$\pm$0.06  &  91.13$\pm$0.81  &  82.71$\pm$0.46  & 92.51 \\
			Forward &  91.66$\pm$0.24  &  90.08$\pm$0.19  &  87.61$\pm$0.44  &  83.26$\pm$0.40  &  78.29$\pm$0.70  & 86.18 &  95.23$\pm$0.14  &  93.27$\pm$0.24  &  89.40$\pm$0.23  &  84.00$\pm$0.37  &  76.68$\pm$0.66  & 87.72 \\
			Reweight &  90.86$\pm$0.26  &  88.79$\pm$0.29  &  84.89$\pm$0.42  &  78.83$\pm$0.93  &  70.41$\pm$2.20  & 82.76 &  94.33$\pm$0.15  &  91.73$\pm$0.46  &  86.68$\pm$0.42  &  79.35$\pm$1.01  &  72.03$\pm$0.77  & 84.82 \\
			BLTM &  92.76$\pm$0.39  &  92.31$\pm$0.64  &  91.97$\pm$0.47  &  89.63$\pm$1.32  &  81.10$\pm$1.90  & 89.55 &  97.04$\pm$0.20  &  96.23$\pm$0.21  &  94.17$\pm$0.48  &  93.59$\pm$1.32  &  86.31$\pm$4.01  & 93.47 \\
			\hline DL-MV &  90.59$\pm$0.33  &  87.75$\pm$1.20  &  81.82$\pm$1.23  &  73.13$\pm$0.70  &  62.29$\pm$0.43  & 79.12 &  93.79$\pm$0.36  &  89.91$\pm$0.42  &  82.06$\pm$0.56  &  71.84$\pm$0.72  &  62.58$\pm$0.29  & 80.04 \\
			DL-CRH &  92.41$\pm$0.10  &  90.67$\pm$0.18  &  88.22$\pm$0.15  &  83.53$\pm$1.09  &  74.74$\pm$2.16  & 85.91 &  96.48$\pm$0.15  &  95.25$\pm$0.06  &  92.10$\pm$0.18  &  89.69$\pm$0.19  &  81.49$\pm$0.36  & 91.00 \\
			DL-MMSR &  92.60$\pm$0.27  &  91.21$\pm$0.21  &  88.85$\pm$0.12  &  82.65$\pm$1.06  &  66.24$\pm$6.54  & 84.31 &  96.37$\pm$0.16  &  95.59$\pm$0.22  &  92.98$\pm$0.07  &  84.33$\pm$0.80  &  70.30$\pm$0.49  & 87.91 \\
			Max-MIG &  92.01$\pm$0.21  &  91.12$\pm$0.28  &  88.78$\pm$0.58  &  83.32$\pm$0.63  &  78.73$\pm$0.25  & 86.79 &  95.80$\pm$0.21  &  93.57$\pm$0.13  &  89.70$\pm$0.39  &  82.70$\pm$0.52  &  74.48$\pm$0.24  & 87.25 \\
			\hline DL-DS &  92.79$\pm$0.08  &  92.35$\pm$0.17  &  91.76$\pm$0.09  &  90.76$\pm$0.12  &  88.92$\pm$0.63  & 91.32 &  96.74$\pm$0.09  &  95.71$\pm$0.03  &  94.20$\pm$0.48  &  90.29$\pm$0.49  &  83.31$\pm$0.81  & 92.05 \\
			DL-IBCC &  92.75$\pm$0.09  &  92.29$\pm$0.13  &  91.48$\pm$0.20  &  89.82$\pm$0.28  &  86.30$\pm$1.00  & 90.53 &  96.75$\pm$0.04  &  95.80$\pm$0.07  &  94.01$\pm$0.12  &  90.89$\pm$0.36  &  83.85$\pm$0.28  & 92.26 \\
			DL-EBCC &  92.70$\pm$0.03  &  92.19$\pm$0.06  &  91.00$\pm$0.12  &  89.07$\pm$0.60  &  86.49$\pm$0.83  & 90.29 &  96.59$\pm$0.19  &  95.64$\pm$0.24  &  94.20$\pm$0.17  &  88.80$\pm$0.25  &  81.85$\pm$0.04  & 91.42 \\
			AggNet &  92.14$\pm$0.10  &  90.98$\pm$0.17  &  89.57$\pm$0.03  &  87.81$\pm$0.59  &  84.74$\pm$0.46  & 89.05 &  94.92$\pm$0.30  &  93.30$\pm$0.11  &  89.78$\pm$0.16  &  84.26$\pm$0.19  &  78.44$\pm$0.30  & 88.14 \\
			CrowdLayer &  92.79$\pm$0.13  &  92.51$\pm$0.22  &  91.64$\pm$0.44  &  89.81$\pm$1.40  &  69.66$\pm$0.17  & 87.28 &  97.07$\pm$0.13  &  96.87$\pm$0.15  &  95.96$\pm$0.10  &  93.99$\pm$0.35  &  89.27$\pm$0.22  & 94.30 \\
			MBEM &  92.12$\pm$0.13  &  91.38$\pm$0.22  &  90.37$\pm$0.40  &  88.66$\pm$0.11  &  78.10$\pm$5.79  & 88.13 &  94.92$\pm$0.30  &  93.30$\pm$0.11  &  89.78$\pm$0.16  &  84.26$\pm$0.19  &  78.44$\pm$0.30  & 88.14 \\
			UnionNet &  92.87$\pm$0.11  &  92.67$\pm$0.14  &  92.00$\pm$0.11  &  90.60$\pm$0.64  &  84.93$\pm$5.36  & 90.61 &  \best{97.45$\pm$0.10}  &  \best{97.05$\pm$0.22}  &  96.06$\pm$0.21  &  93.46$\pm$0.16  &  88.33$\pm$0.16  & 94.47 \\
			CoNAL &  91.75$\pm$0.19  &  90.16$\pm$0.08  &  86.32$\pm$0.44  &  81.13$\pm$0.99  &  75.47$\pm$1.62  & 84.97 &  95.80$\pm$0.08  &  92.68$\pm$0.26  &  85.94$\pm$0.37  &  78.01$\pm$0.25  &  69.92$\pm$0.51  & 84.47 \\
			\hline TAIDTM &  \best{93.27$\pm$0.19}  &  \best{92.98$\pm$0.13}  &  \best{92.86$\pm$0.18}  &  \best{92.22$\pm$0.58}  &  \best{89.62$\pm$0.94}  & \best{92.19} &  97.22$\pm$0.16  &  96.71$\pm$0.16  &  \best{96.15$\pm$0.44}  &  \best{95.62$\pm$0.45}  &  \best{90.18$\pm$3.87}  & \best{95.18} \\
			\hline
		\end{tabular}
		\label{fmnist_kmnist}
	\end{table*}
	\begin{table*}[ht]
		\caption{Comparison with state-of-the-art methods in the test accuracy (\%) on simulated CIFAR10 and SVHN datasets. The best results are in \best{bold}.}
		\centering
		\fontsize{7.5pt}{10pt}\selectfont
		\setlength\tabcolsep{2.5pt}
		\begin{tabular}{c|ccccc|c||ccccc|c}
			\hline \multirow{2}{*}{ Dataset } & \multicolumn{6}{c||}{ CIFAR10 } & \multicolumn{6}{c}{ SVHN } \\
			\cline{2-13} 
			& AIDN-10\% & AIDN-20\% & AIDN-30\% & AIDN-40\% & AIDN-50\% & Avg. & AIDN-10\% & AIDN-20\% & AIDN-30\% & AIDN-40\% & AIDN-50\% & Avg. \\
			\hline   CE &  77.13$\pm$0.35  &  72.94$\pm$0.28  &  68.03$\pm$0.58  &  60.33$\pm$0.77  &  52.33$\pm$0.92  & 66.15 &  91.54$\pm$0.07  &  88.30$\pm$0.32  &  83.22$\pm$0.29  &  75.71$\pm$0.25  &  67.41$\pm$0.26  & 81.24 \\
			GCE &  82.32$\pm$0.28  &  81.35$\pm$0.31  &  78.75$\pm$0.22  &  72.33$\pm$0.50  &  62.32$\pm$0.38  & 75.41 &  95.12$\pm$0.01  &  94.60$\pm$0.21  &  93.23$\pm$0.14  &  86.99$\pm$0.28  &  76.11$\pm$1.14  & 89.21 \\
			Forward &  78.19$\pm$0.44  &  75.11$\pm$0.33  &  71.19$\pm$0.96  &  65.47$\pm$0.40  &  57.63$\pm$2.32  & 69.52 &  92.65$\pm$0.15  &  90.44$\pm$0.26  &  87.14$\pm$0.72  &  83.02$\pm$0.89  &  79.07$\pm$0.89  & 86.46 \\
			Reweight &  79.22$\pm$0.44  &  75.66$\pm$0.39  &  71.86$\pm$0.69  &  65.11$\pm$0.60  &  60.39$\pm$1.77  & 70.45 &  92.22$\pm$0.14  &  89.37$\pm$0.41  &  85.27$\pm$0.15  &  78.71$\pm$0.81  &  72.46$\pm$1.55  & 83.61 \\
			BLTM &  82.70$\pm$0.64  &  81.03$\pm$0.59  &  78.93$\pm$0.22  &  73.06$\pm$1.86  &  60.51$\pm$1.67  & 75.25 &  95.01$\pm$0.34  &  95.00$\pm$0.13  &  93.68$\pm$0.36  &  92.60$\pm$0.99  &  80.59$\pm$3.96  & 91.38 \\
			\hline
			DL-MV &  79.29$\pm$0.52  &  74.09$\pm$0.34  &  67.15$\pm$0.35  &  56.82$\pm$1.50  &  46.96$\pm$1.55  & 64.86 &  92.55$\pm$0.45  &  88.74$\pm$0.25  &  83.15$\pm$0.55  &  73.99$\pm$0.15  &  64.28$\pm$1.02  & 80.54 \\
			DL-CRH &  83.06$\pm$0.15  &  80.36$\pm$0.34  &  75.64$\pm$0.62  &  67.52$\pm$0.94  &  52.66$\pm$5.99  & 71.85 &  94.47$\pm$0.08  &  93.06$\pm$0.08  &  89.04$\pm$0.26  &  80.40$\pm$0.12  &  67.43$\pm$1.42  & 84.88 \\
			DL-MMSR &  83.06$\pm$0.33  &  80.66$\pm$0.26  &  75.57$\pm$0.47  &  63.35$\pm$0.20  &  50.02$\pm$2.30  & 70.53 &  94.84$\pm$0.14  &  93.75$\pm$0.39  &  91.12$\pm$0.31  &  83.96$\pm$3.11  &  64.89$\pm$1.98  & 85.71 \\
			Max-MIG &  82.53$\pm$0.62  &  80.10$\pm$0.33  &  76.38$\pm$0.66  &  70.90$\pm$0.87  &  65.16$\pm$0.15  & 75.01 &  92.65$\pm$0.59  &  91.37$\pm$0.24  &  88.76$\pm$0.16  &  84.26$\pm$0.49  &  78.64$\pm$1.80  & 87.14 \\
			\hline
			DL-DS &  83.32$\pm$0.34  &  82.34$\pm$0.35  &  79.82$\pm$0.28  &  78.01$\pm$0.69  &  72.68$\pm$1.44  & 79.23 &  94.55$\pm$0.03  &  93.22$\pm$0.12  &  88.98$\pm$0.69  &  80.66$\pm$0.53  &  66.22$\pm$1.58  & 84.73 \\
			DL-IBCC &  83.51$\pm$0.15  &  82.00$\pm$0.13  &  79.78$\pm$0.34  &  75.82$\pm$0.27  &  69.77$\pm$0.13  & 78.18 &  95.17$\pm$0.10  &  94.67$\pm$0.08  &  93.92$\pm$0.16  &  92.07$\pm$0.70  &  86.27$\pm$5.23  & 92.42 \\
			DL-EBCC &  83.62$\pm$0.46  &  82.53$\pm$0.48  &  80.25$\pm$0.29  &  77.88$\pm$0.61  &  74.50$\pm$0.77  & 79.76 &  95.16$\pm$0.06  &  94.68$\pm$0.04  &  93.76$\pm$0.14  &  91.47$\pm$0.89  &  87.65$\pm$3.01  & 92.54 \\
			AggNet &  81.80$\pm$0.42  &  79.05$\pm$0.42  &  77.17$\pm$0.29  &  74.50$\pm$0.65  &  72.28$\pm$0.45  & 76.96 &  94.21$\pm$0.27  &  93.61$\pm$0.08  &  92.52$\pm$0.37  &  91.49$\pm$0.52  &  91.19$\pm$0.34  & 92.60 \\
			CrowdLayer &  \best{83.82$\pm$0.71}  &  82.19$\pm$0.35  &  77.43$\pm$3.67  &  72.14$\pm$4.01  &  58.33$\pm$10.18  & 74.78 &  95.11$\pm$0.61  &  95.22$\pm$0.18  &  90.42$\pm$5.75  &  88.06$\pm$5.88  &  78.33$\pm$2.22  & 89.43 \\
			MBEM &  80.66$\pm$0.45  &  77.20$\pm$0.51  &  70.58$\pm$2.44  &  67.61$\pm$0.78  &  59.83$\pm$1.06  & 71.18 &  94.04$\pm$0.40  &  93.48$\pm$0.06  &  91.82$\pm$0.58  &  90.49$\pm$0.41  &  82.41$\pm$4.59  & 90.45 \\
			UnionNet &  83.32$\pm$0.39  &  81.64$\pm$0.18  &  79.52$\pm$0.46  &  76.24$\pm$0.58  &  66.68$\pm$2.90  & 77.48 &  \best{95.44$\pm$0.04}  &  94.91$\pm$0.14  &  91.75$\pm$3.49  &  87.16$\pm$5.48  &  70.88$\pm$4.95  & 88.03 \\
			CoNAL &  82.10$\pm$0.22  &  79.22$\pm$0.10  &  74.78$\pm$0.27  &  68.57$\pm$0.53  &  62.39$\pm$0.76  & 73.41 &  94.11$\pm$0.11  &  92.41$\pm$0.23  &  89.24$\pm$0.43  &  84.59$\pm$0.41  &  73.70$\pm$3.39  & 86.81 \\
			\hline
			TAIDTM &  83.28$\pm$0.95  &  \best{82.72$\pm$0.69}  &  \best{82.14$\pm$0.39}  &  \best{79.95$\pm$0.35}  &  \best{77.18$\pm$0.62}  & \best{81.05} &  95.17$\pm$0.28  &  \best{95.34$\pm$0.04}  &  \best{94.80$\pm$0.18}  &  \best{94.65$\pm$0.14}  &  \best{94.48$\pm$0.42}  & \best{94.89} \\
			\hline
		\end{tabular}
		\label{cifar_svhn}
	\end{table*}
	
	\section{Experiments}
	\label{exp}
	\myPara{Baselines.}
	To make evaluations comprehensive, we exploit three types of baselines in experiments. Type-I baselines are the methods that treat all noisy labels from the same noisy distribution. 
	Type-II baselines are model-free methods designed for learning from crowds, which do not explicitly model the noise generation process. 
	Type-III baselines are model-based algorithms designed for learning from crowds, which model the noise generation process of each annotator by noise transition matrices. Specifically, the Type I baselines include CE, GCE~\cite{zhang2018generalized}, Forward~\cite{PatriniRMNQ17}, Reweight~\cite{Liu2016TPAMI}, 
	and BLTM~\cite{yang2021estimating}. The Type II baselines
	include DL-MV, DL-CRH~\cite{LiLGZFH14}, DL-MMSR~\cite{MaO20}, and Max-MIG~\cite{CaoXKW19}. The Type III baselines include DL-DS~\cite{Dawid1979Maximum}, DL-IBCC~\cite{KimG12}, DL-EBCC~\cite{LiRC19}, AggNet~\cite{Albarqouni2016TMI}, CrowdLayer~\cite{Rodrigues2018aaai}, MBEM~\cite{Khetan2018iclr}, UnionNet~\cite{Wei2022DeepLF}, and CoNAL~\cite{Chu0W21}. We detail all baselines in Appendix~\ref{detail}. 
	\subsection{Evaluations on Simulated Datasets	}
	\myPara{Datasets.}
	\label{data}
	We conduct experiments on four simulated datasets to verify the effectiveness of our method, \ie, Fashion-MNIST (F-MNIST)~\cite{Xiao2017FashionMNISTAN}, Kuzushiji-MNIST (K-MNIST)~\cite{Clanuwat2018DeepLF}, CIFAR10~\cite{Krizhevsky2009LearningML}, and SVHN~\cite{Netzer2011ReadingDI}. F-MNIST~\cite{Xiao2017FashionMNISTAN} and K-MNIST~\cite{Clanuwat2018DeepLF} have 10 classes of 28 $\times$ 28 grayscale images, including 60,000 training images and 10,000 test images. CIFAR10~\cite{Krizhevsky2009LearningML} has 10 classes of 32 $\times$ 32 $\times$ 3 images, including 50,000 training images and 10,000 test images. SVHN~\cite{Netzer2011ReadingDI} has 10 classes of 32 $\times$ 32 $\times$ 3 images with 73,257 training images and 26,032 test images. For these datasets, we leave out 10\% of training examples as a validation set. The four datasets contain clean data. To simulate noisy labels for each annotator, following~\cite{yang2021estimating}, we corrupt training and validation sets manually with instance-dependent noise (IDN) according to the synthetic instance-dependent transition matrices. We generate 300 annotators in these datasets by dividing them into 3 groups with the same label noise patterns, which means every 100 annotators share the same instance-dependent transition matrices. Each annotator randomly chooses instances to label, and each example has $\bar{r}=2$ labels on average. More details can be found in Appendix~\ref{gener_label_noise}. Note that we perform the experiments under different individual noise rate $\rho\%$, and ``AIDN-$\rho\%$'' means to generate multiple annotators whose noise rate is $\rho\%$ and noise type is ``IDN".
	
	\myPara{Implementation details.} In the experiments on simulated datasets, we use ResNet-18~\cite{he2016deep} for F-MNIST and K-MNIST, and ResNet-34 networks~\cite{he2016deep} for CIFAR10 and SVHN. The noise-transition networks are
	the same architecture as the classification network, but the last linear layer is modified according to
	the transition matrix shape. The number of GCN layers is 2. We use SGD with a momentum 0.9, batch
	size 128, and a learning rate of 0.01 to learn the noise-transition networks. The warmup epoch in distilled examples collection~\cite{yang2021estimating} is 5 for F-MNIST, K-MNIST, and CIFAR10, 10 for SVHN. The training epoch in learning the global noise-transition network is 5 for F-MNIST, K-MNIST, and CIFAR10, 10 for SVHN. The fine-tuning epoch to learn the individual noise-transition networks is 2 for all datasets, and $k$ is set to 50.  The training epoch in the knowledge transfer between neighboring individuals is 10 epochs for F-MNIST, K-MNIST and CIFAR10, 15 epochs for SVHN. The classification network is trained on the noisy dataset for 60 epochs for all datasets using SGD optimizer with an initial learning rate of 0.01 and weight decay of 1e-4, and the learning rate is divided by 10 after 40 and 55 epochs, respectively. Note that, for a fair comparison, we do not use any data augmentation technique in all experiments as in~\cite{xia2019anchor,yang2021estimating}. We use the test accuracy of the classifier achieved in the final training epoch for evaluations. All experiments are repeated 3 times. The average and standard deviation values of the results are reported. 
	
	\myPara{Experimental results.}
	Table~\ref{fmnist_kmnist} reports the test accuracy on F-MNIST and K-MNIST datasets, and Table~\ref{cifar_svhn} reports the test accuracy on CIFAR10 and SVHN datasets. In summary, our method (named TAIDTM) gets the highest average performance across various noise settings. Below, we further discuss the results based on the comparisons with three different types of baselines.
	
	For Type-I baselines, we first notice that among them, BLTM gets the competitive classification performance in most cases due to its effective instance-dependent modeling, which clearly illustrates the necessity to model instance-dependent noise patterns. Additionally, we compare our TAIDTM with BLTM. Without considering the difference of annotators, TAIDTM will reduce to BLTM. As shown in the reported performances, TAIDTM significantly performs much better than BLTM, and the superiority of TAIDTM is gradually revealed along with the noise rate increasing. Specifically, on F-MNIST and K-MNIST datasets with AIDN-50$\%$, TAIDTM brings +8.52\% and +3.87\% improvements, respectively. Also, on CIFAR10 and SVHN datasets with AIDN-50\%, TAIDTM brings +16.67\% and +13.89\% improvements, respectively.
	
	For Type-II baselines, we can see that although by considering the different reliability of annotators, DL-CRH, DL-MMSR and Max-MIG achieve improvements compared with DL-MV, which naively assuming all annotators have the same labeling accuracy, they can not handle different instance- dependent noisy annotators well. When compared with them, our TAIDTM consistently performs better with a large margin across various noise rates, especially for high noise rate cases.
	
	For Type-III baselines, none of these methods can overall perform better than others across various cases on all four datasets. Except for CoNAL, their assumption that the noise generation is annotator-dependent but instance-independent is not satisfied when faced with instance-dependent label noise, and thus their performance cannot be guaranteed. As for CoNAL, it decomposes annotation noise into common noise and individual noise, and it can be regarded as estimating a special AIDTM, where one noise pattern is shared among all annotators. While, in real-world complex data, various noisy patterns may be shared by different annotators. In contrast to these methods, our method supposes various noisy patterns are shared among similar annotators, and it exploits deep networks to explicitly model such general label noise, hence leading to better average performance.

	\subsection{Evaluations on Real-World Datasets}
	\myPara{Datasets.}
	We conduct experiments on three real-world datasets to verify the effectiveness of our method, \ie, LabelMe~\cite{Rodrigues2018aaai}, Music~\cite{Rodrigues2014GaussianPC}, and CIFAR10-N~\cite{Wei2022LearningWN}. LabelMe~\cite{Rodrigues2018aaai} is a real-world 8-class image classification dataset. It consists of 2,688 images, where 1,000 of them are used to obtain noisy labels by an average of 2.5 annotators per image (59 annotators in total) from Amazon Mechanical Turk. 500 images are used for validation, while 1108 images are used for testing. 
	We follow the image pre-processing method in~\cite{Rodrigues2018aaai}. Music~\cite{Rodrigues2014GaussianPC} is a music genre classification dataset, consisting of 1,000 examples of
	songs with 30 seconds length from 10 music genres, where
	700 of them are labeled by Amazon-Mechanical-Turk annotators and the rest is used for testing. Each example is labeled by an average of 4.2
	annotators.  
	CIFAR10-N~\cite{Wei2022LearningWN} is a variant of CIFAR10 with human-annotated real-world noisy labels from Amazon Mechanical Turk. Each example is labeled by 3 annotators (747 annotators in total). 10\% of training images are left out for validation. Test images of the original CIFAR10 dataset are used for testing.
	
	\myPara{Implementation Details.} For LabelMe, we apply a pre-trained VGG-16 network
	followed by an FC layer with 128 units, ReLU activations, and a softmax output layer. The number
	of GCN layers is 1. For Music, we use the same FC layer and softmax layer as LabelMe, where
	batch normalization is performed. The number of GCN layers is 2. We use SGD with a momentum 0.9,
	batch size 128, and a learning rate of 0.01 to learn the noise-transition and classification networks.
	The warmup epoch is 10 for LabelMe and 50 for Music. The training epoch in learning the global
	noise-transition network is 10 for LabelMe and 20 for Music. The fine-tuning epoch and $k$ are set
	to 1. The number of training epochs in knowledge transfer between neighboring individuals is 2 for LabelMe
	and 40 for Music. The number of training epochs in the learning classifier network with loss correction is 50. For
	CIFAR10-N, we use the same networks and hyper-parameters with the simulated CIFAR-10 dataset.
	We use the test accuracy of the classifier achieved in the final training epoch for evaluations. The
	experiments on CIFAR10-N are repeated 3 times, and LabelMe and Music are repeated 50 times.
		\begin{table}
		\caption{Comparison with state-of-the-art methods in the test accuracy (\%) on real-world datasets. The best results are in \best{bold}.}
		\centering
		\fontsize{8pt}{10pt}\selectfont
		\setlength\tabcolsep{3.5pt}
		\begin{tabular}{c|ccc|c}
			\hline Dataset & LabelMe & Music & CIFAR10-N & Avg. \\
			\hline CE &  80.02$\pm$0.91  &  62.24$\pm$1.08  &  80.25$\pm$0.17  & 74.17 \\
			GCE &  80.91$\pm$1.21  &  66.99$\pm$1.20  &  83.02$\pm$0.40  & 76.97 \\
			Forward &  80.52$\pm$1.58  &  68.40$\pm$3.06  &  79.74$\pm$0.78  & 76.22 \\
			Reweight &  79.49$\pm$1.31  &  68.78$\pm$2.22  &  81.32$\pm$0.73  & 76.53 \\
			BLTM &  85.63$\pm$0.55  &  66.28$\pm$2.23  &  83.25$\pm$0.53  & 77.03 \\
			\hline DL-MV &  77.92$\pm$0.63  &  64.10$\pm$1.49  &  80.96$\pm$0.49  & 78.39 \\
			DL-CRH &  78.30$\pm$0.76  &  68.73$\pm$1.39  &  80.78$\pm$0.35  & 75.94 \\
			DL-MMSR &  78.17$\pm$1.13  &  65.44$\pm$1.34  &  73.09$\pm$0.66  & 72.23 \\
			Max-MIG &  83.71$\pm$0.39  &  \best{70.71}$\pm$\best{1.75}  &  83.04$\pm$1.04  & 79.15 \\
			\hline DL-DS &  80.29$\pm$1.08  &  67.79$\pm$1.62  &  80.71$\pm$0.13  & 76.26 \\
			DL-IBCC &  80.74$\pm$0.82  &  67.10$\pm$1.57  &  82.42$\pm$0.68  & 76.75 \\
			DL-EBCC &  80.42$\pm$0.69  &  67.77$\pm$1.58  &  82.50$\pm$0.13  & 76.90 \\
			AggNet &  82.69$\pm$0.49  &  67.59$\pm$1.68  &  82.97$\pm$0.11  & 77.75 \\
			CrowdLayer &  82.55$\pm$0.72  &  67.93$\pm$1.44  &  83.49$\pm$0.26  & 77.99 \\
			MBEM &  80.20$\pm$2.14  &  65.53$\pm$1.45  &  82.13$\pm$0.34  & 75.95 \\
			UnionNet &  83.17$\pm$1.66  &  68.62$\pm$1.31  &  82.56$\pm$0.12  & 78.12 \\
			CoNAL &  84.85$\pm$0.91  &  68.67$\pm$1.88  &  82.83$\pm$0.28  & 78.78 \\
			\hline TAIDTM &  \best{86.68}$\pm$\best{0.49}  &  69.39$\pm$1.96  &  \best{83.99}$\pm$\best{0.51}  &  \best{80.02}  \\
			\hline
		\end{tabular}
		\label{real}
	\end{table}	
	\begin{figure*}[ht]
		\centering
		\subfloat{\includegraphics[width=3.95in]{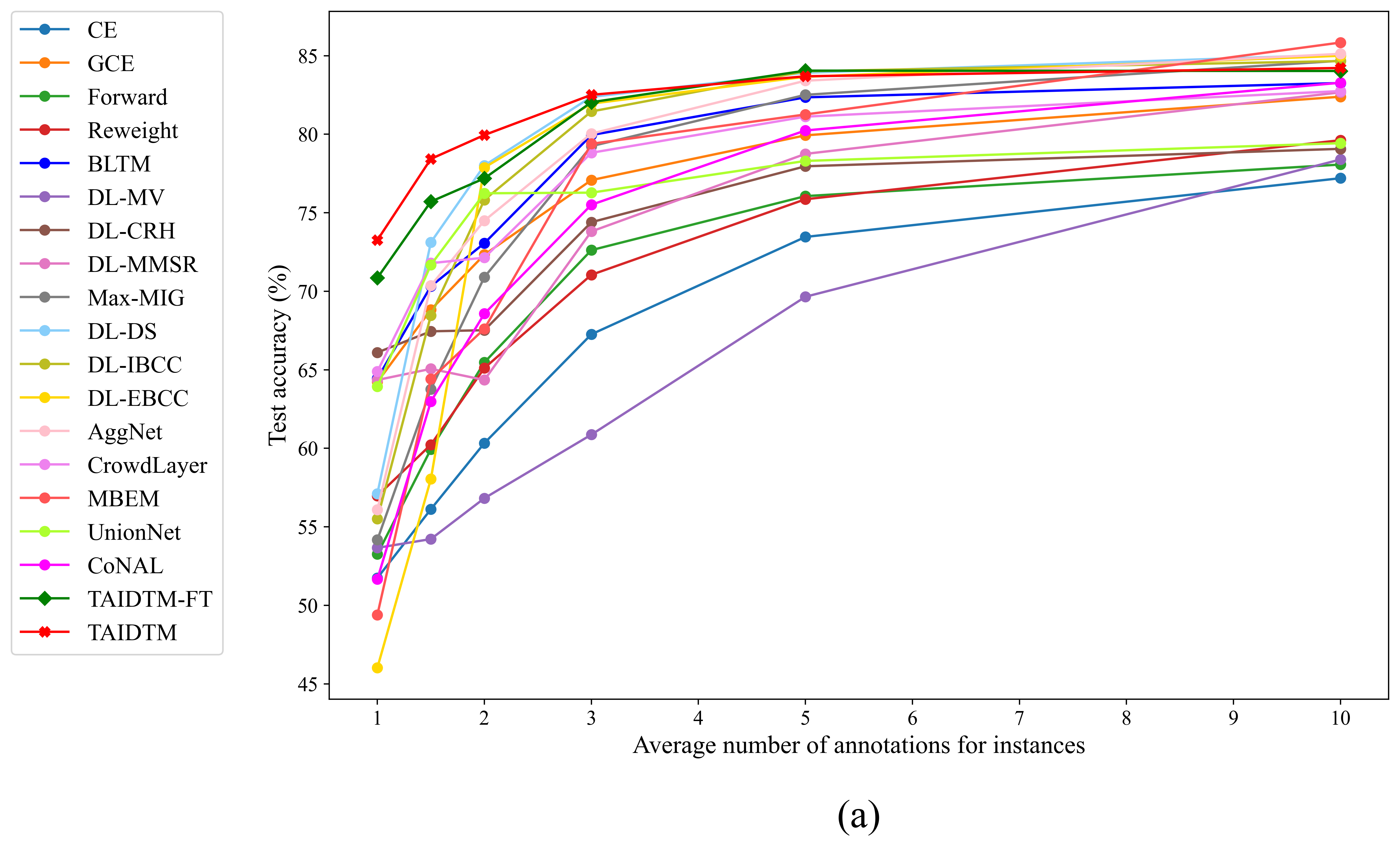}\label{fig_sparsity}}
		\hfil
		\subfloat{\includegraphics[width=3.2in]{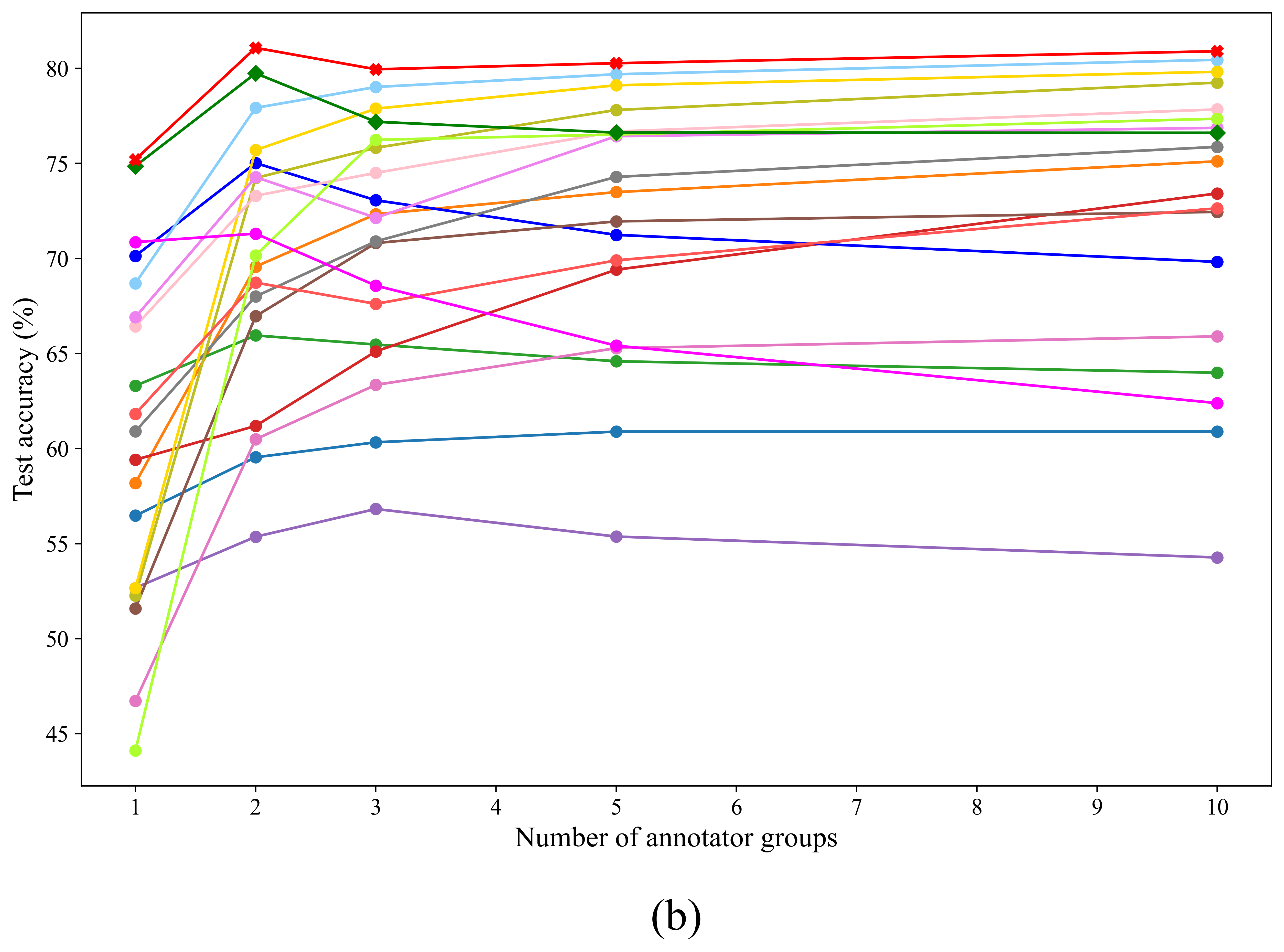}\label{fig_relevance}}
		\caption{ The ablation study conducted on the CIFAR10 dataset with AIDN-40\% label noise. (a) the test accuracy (\%) \textit{vs.} the average number of annotations for instances;
			(b) the test accuracy (\%) \textit{vs.} the number of annotator groups. } 
		\label{ablation}
	\end{figure*}
	\begin{table*}[h]
		\caption{Ablation study about knowledge transfer. "X-40\%" represents the $\mathrm{X}$ dataset with AIDN-40\% label noise. The best results are in \best{bold}.}
		\centering
		\fontsize{8.0pt}{10pt}\selectfont
		\setlength\tabcolsep{4.4pt}
		\label{role}
		\begin{tabular}{l|ccccccc}
			\hline Datasets & F-MNIST-40\% & K-MNIST-40\% & CIFAR10-40\% & SVHN-40\% & LabelMe & Music & CIFAR10-N \\
			\hline BLTM &  89.63$\pm$1.32  &  93.59$\pm$1.32  &  73.06$\pm$1.86  &  92.60$\pm$0.99  &  86.16$\pm$0.41  &  68.56$\pm$1.62  &  83.84$\pm$0.15  \\
			TAIDTM-FT &  89.58$\pm$2.39  &  95.39$\pm$0.50  &  77.19$\pm$1.38  &  93.93$\pm$0.76  &  85.36$\pm$0.78  &  67.23$\pm$1.33  &  83.90$\pm$0.23  \\
			TAIDTM &  \best{92.22}$\pm$\best{0.58}  &  \best{95.62}$\pm$\best{0.45}  &  \best{79.95}$\pm$\best{0.35}  &  \best{94.65}$\pm$\best{0.14}  &  \best{86.68}$\pm$\best{0.49}  &  \best{69.39}$\pm$\best{1.96}  &  \best{83.99}$\pm$\best{0.51}  \\
			\hline
		\end{tabular}
	\end{table*}	
	\myPara{Experimental results.} 
	\label{real2}
	Table~\ref{real} reports the test accuracy on the datasets of LabelMe, Music, and CIFAR10-N. We can first find that our TAIDTM achieves superior performance than the model-based learning-from-crowds methods on all real-world datasets. Second, TAIDTM gets the best results on LabelMe and CIFAR10-N datasets. Especially for the LabelMe dataset, it not only achieves +1.05\% improvement compared with the second-best result but outperforms all learning-from-crowds baselines by a large margin (at least +1.85\%). These results prove our TAIDTM can effectively handle the real-world label noise in learning from crowds. 
	
	\begin{figure*}[!t]
		\begin{minipage}[c]{0.05\textwidth}\centering\small \rotatebox[origin=c]{90}{\tiny{Before Graph Purifying}} \end{minipage}%
		\begin{minipage}[c]{0.95\textwidth}
			\includegraphics[width=0.25\textwidth]{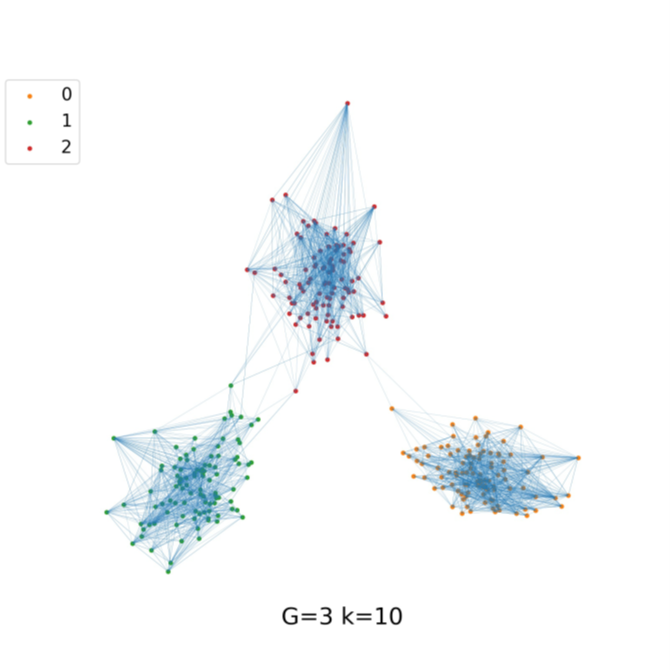}%
			\includegraphics[width=0.25\textwidth]{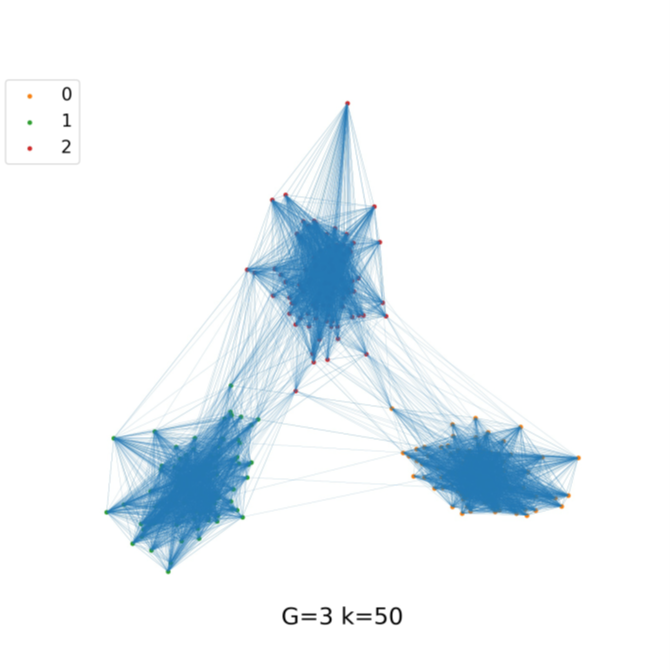}%
			\includegraphics[width=0.25\textwidth]{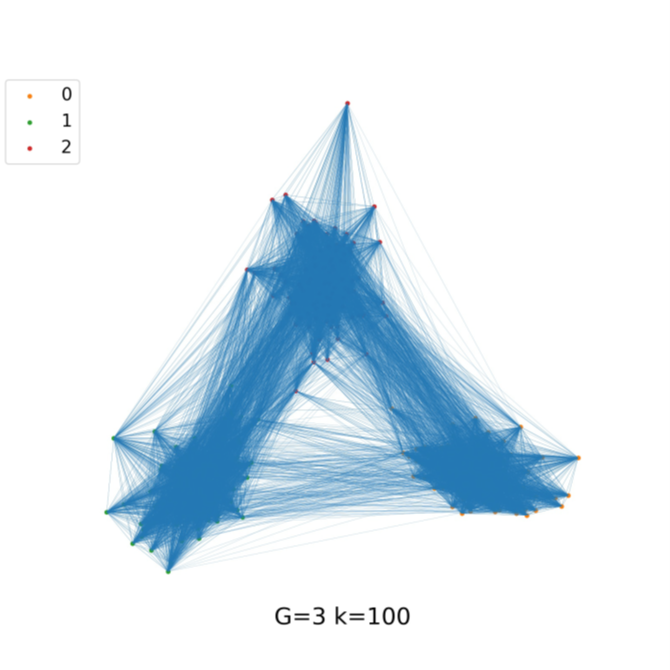}%
			\includegraphics[width=0.25\textwidth]{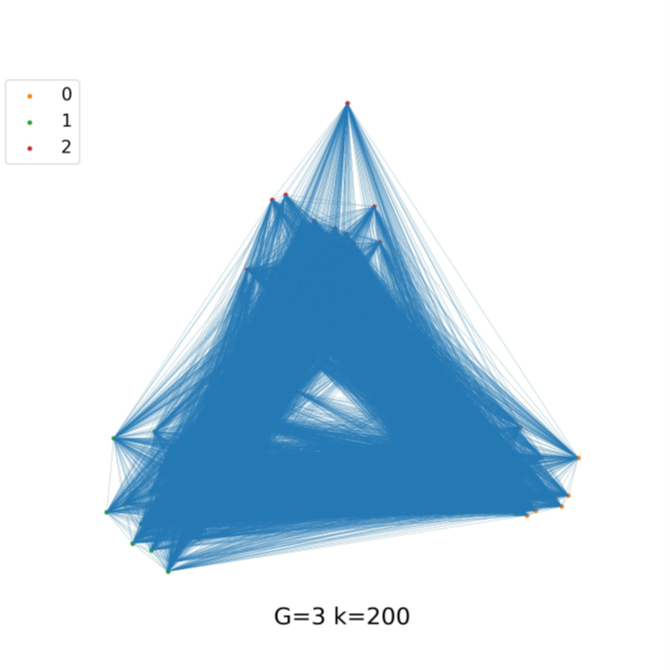}%
		\end{minipage}
		\begin{minipage}[c]{0.05\textwidth}\centering\small \rotatebox[origin=c]{90}{\tiny{After Graph Purifying}} \end{minipage}%
		\begin{minipage}[c]{0.95\textwidth}
			\includegraphics[width=0.25\textwidth]{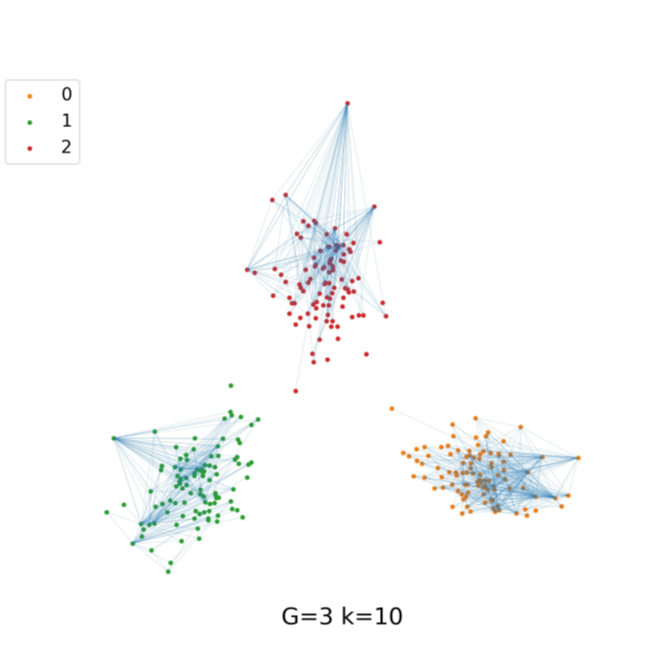}%
			\includegraphics[width=0.25\textwidth]{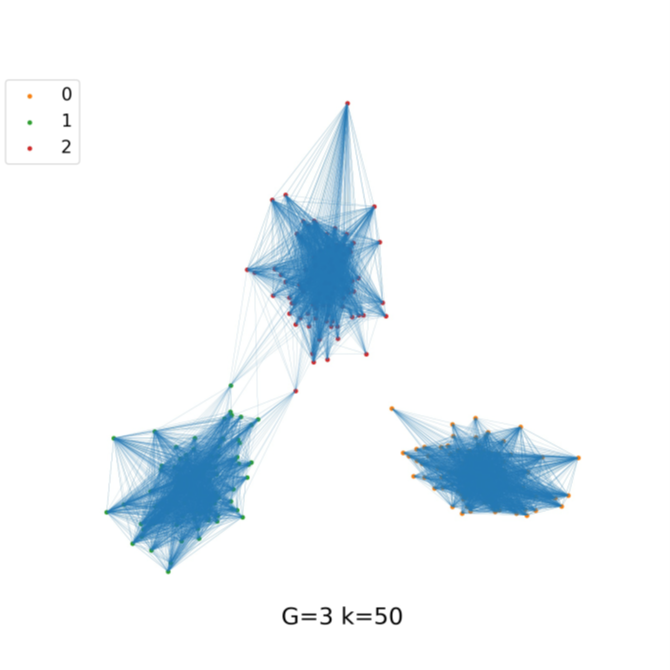}%
			\includegraphics[width=0.25\textwidth]{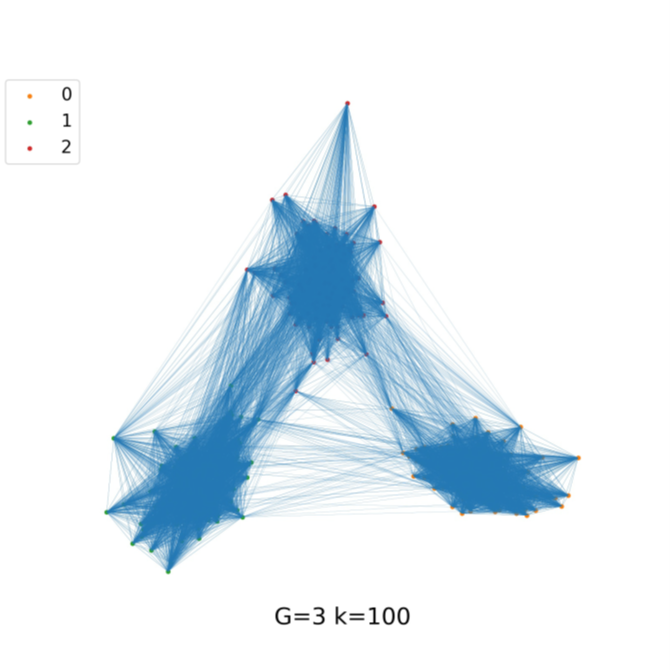}%
			\includegraphics[width=0.25\textwidth]{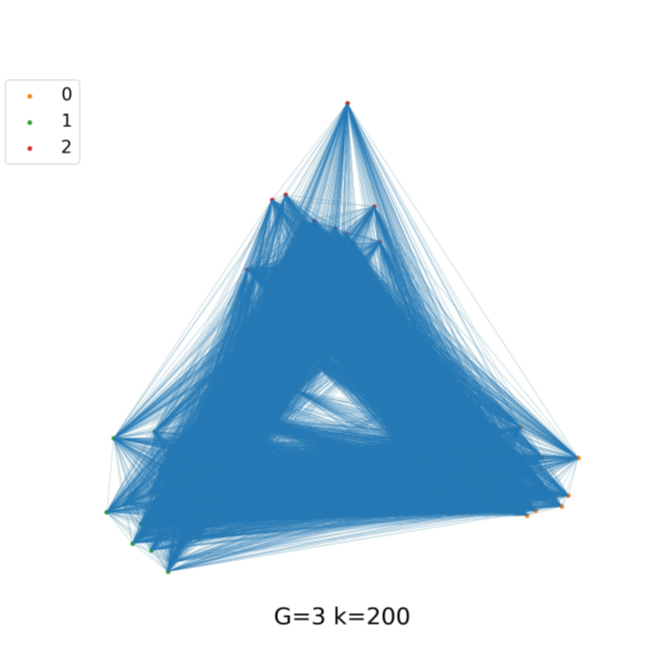}%
		\end{minipage}
		\begin{minipage}[c]{0.05\textwidth}\centering\small \rotatebox[origin=c]{90}{\tiny{Before Graph Purifying}} \end{minipage}%
		\begin{minipage}[c]{0.95\textwidth}
			\includegraphics[width=0.25\textwidth]{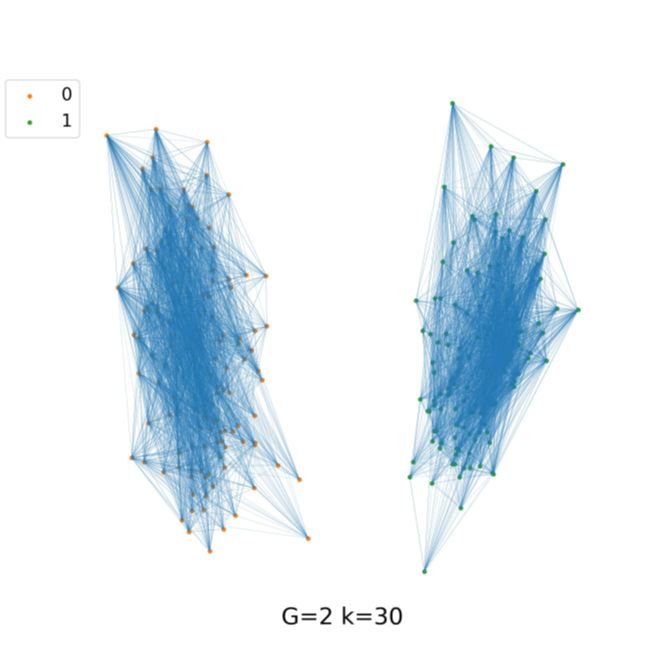}%
			\includegraphics[width=0.25\textwidth]{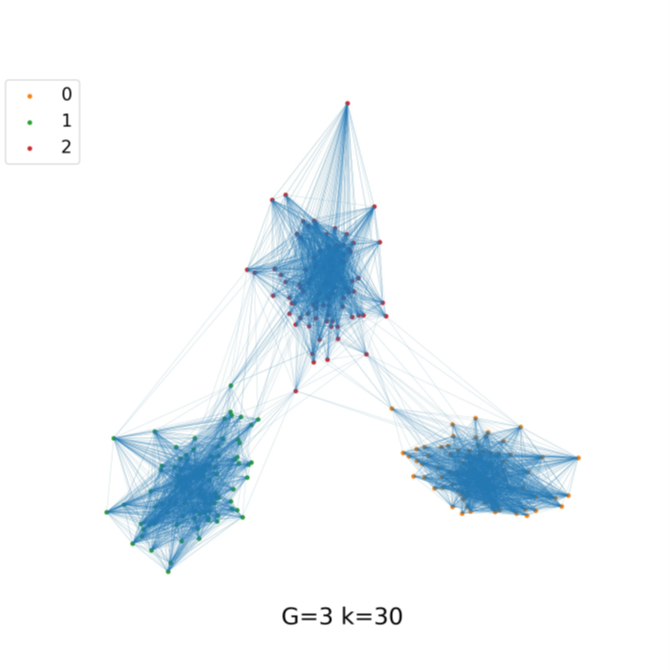}%
			\includegraphics[width=0.25\textwidth]{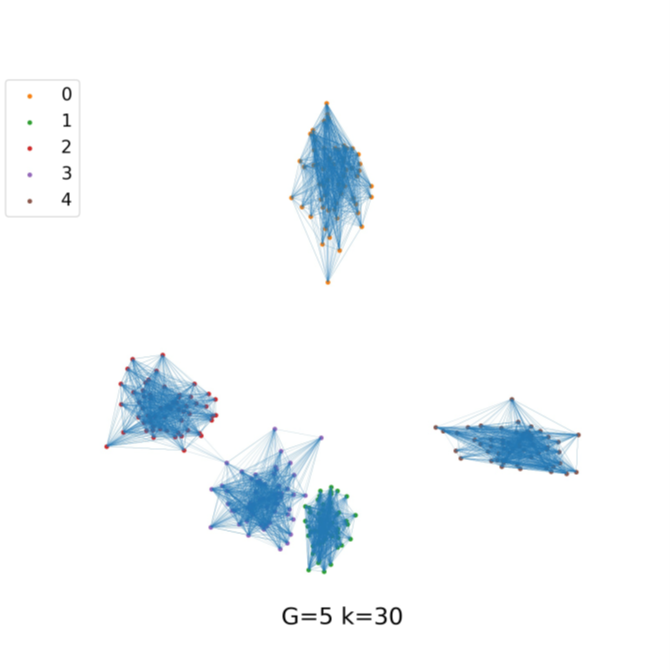}%
			\includegraphics[width=0.25\textwidth]{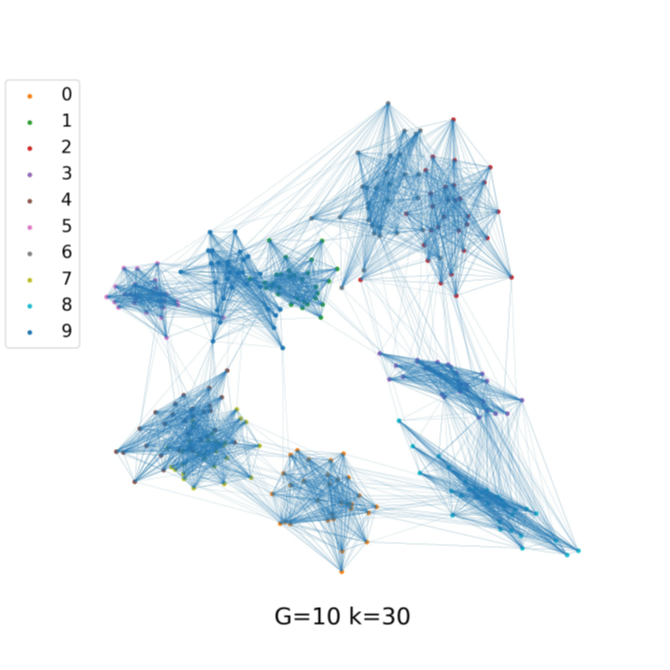}%
		\end{minipage}
		\begin{minipage}[c]{0.05\textwidth}\centering\small \rotatebox[origin=c]{90}{\tiny{After Graph Purifying}} \end{minipage}%
		\begin{minipage}[c]{0.95\textwidth}
			\includegraphics[width=0.25\textwidth]{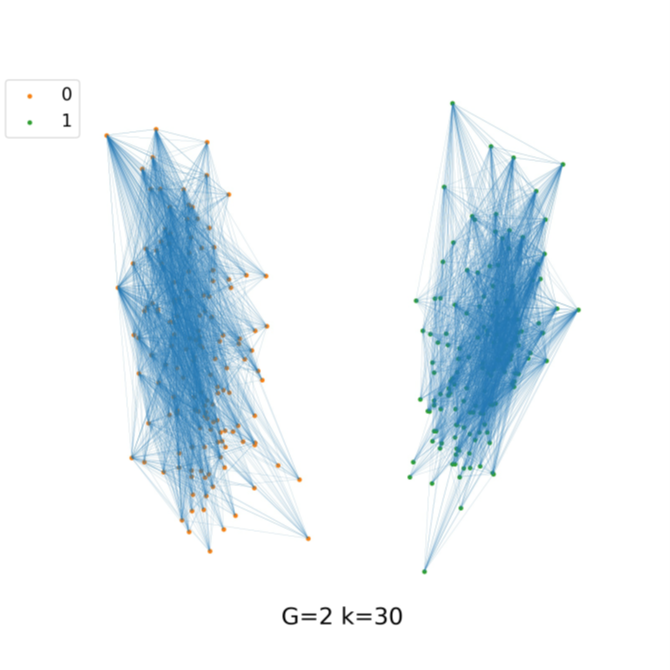}%
			\includegraphics[width=0.25\textwidth]{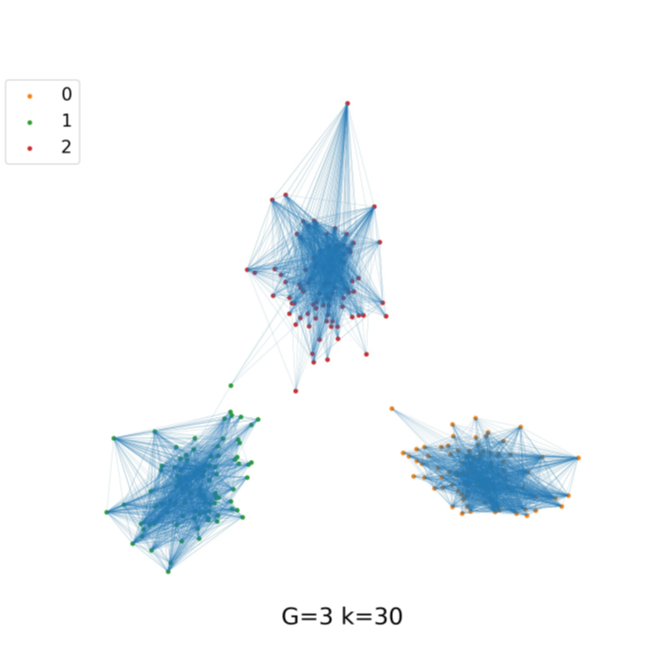}%
			\includegraphics[width=0.25\textwidth]{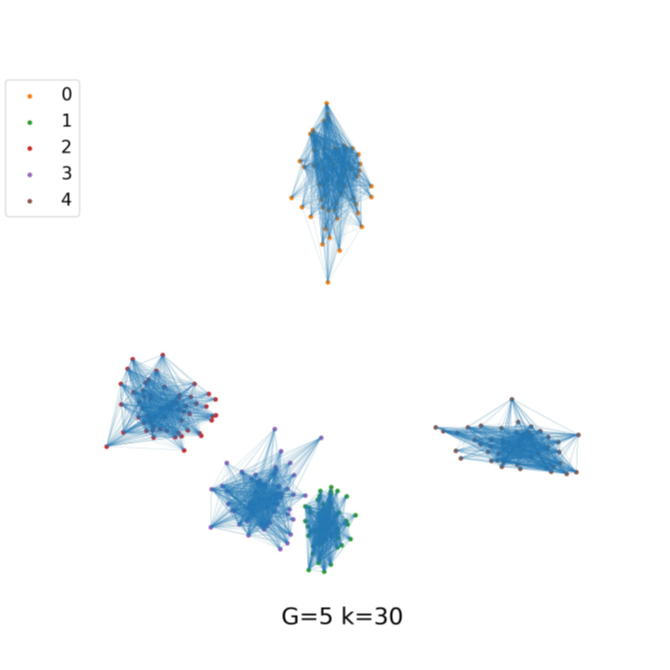}%
			\includegraphics[width=0.25\textwidth]{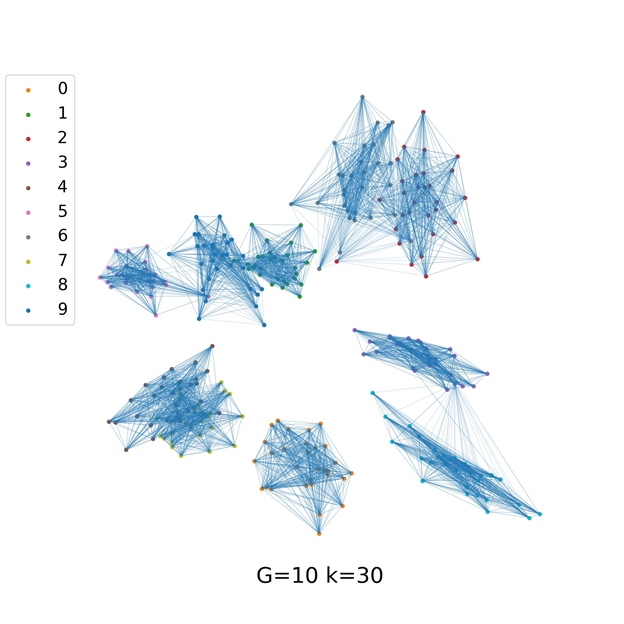}%
		\end{minipage}
		\caption{Visualization of the last layer’s parameters belonging to different individuals and the
			contrusted similarity graph. The first two rows denote that the experiments are conducted under
			diffetent hyperparameters $k$. The last two rows denote that the experiments are conducted under
			different numbers of annotator groups. All experiments are conducted on CIFAR10 with
			AIDN-40$\%$ label noise.}
		\label{fig_v}
	\end{figure*}
	\subsection{Ablation Study}
	\label{data2}
	\myPara{Role of knowledge transfer.} Our approach has
	two knowledge transfers, both of which play an important role in estimating general AIDTM in practice. First, we transfer the global noise-transition
	network to individual noise-transition networks, which addresses the issue that sparse individual
	annotations cannot train a high-complexity neural network. Second, we transfer the knowledge
	between neighboring individuals by a GCN-based mapping function, which addresses the problem
	that the transfer from the mixture of noise patterns to individuals may cause two annotators with
	highly different noise generations to perturb each other. Therefore, without the first knowledge
	transfer, TAIDTM will reduce to BLTM, which trains one-annotator instance-dependent transition matrices. Without the second knowledge transfer, our TAIDTM degenerates to TAIDTM-FT, which learns individual noise-transition networks via fine-tuning. Here we conduct experiments to compare our approach with BLTM and TAIDTM-FT. As shown in Table~\ref{role}, with the help of the first transfer, TAIDTM-FT can improve the performance of BLTM in most cases. Due to the second knowledge
	transfer, TAIDTM consistently achieves superior performance than both BLTM and TAIDTM-FT.
	
	\myPara{Influence of annotation sparsity.} To study the influence of annotation sparsity, we conduct experiments on the CIFAR10 dataset with AIDN-40\% label noise under different average numbers of
	annotations for instances $\bar{r}$. As shown in Fig.~\ref{fig_sparsity}, we can see that the performance of all methods
	increases with the increase of the annotations, and our TAIDTM has better or competitive performance
	compared with other methods across various sparse cases. In addition, although the performance
	of all methods degrades with the sparsity of annotation increasing, the superiority of our TAIDTM
	becomes more significant when annotations are more sparse. Besides, the gap between TAIDTM and
	TAIDTM-FT in Fig.~\ref{fig_sparsity} show the significant improvements through the knowledge transfer between
	neighboring individuals, especially for the highly sparse cases.
	
	\myPara{Influence of annotator relevance.} To study the influence of annotator relevance, we conduct
	experiments on the CIFAR10 dataset with AIDN-40\% label noise under different numbers of the annotator
	groups where annotators share the same instance-dependent transition matrices. The more the
	annotator groups, the less annotator relevance. As shown in Fig.~\ref{fig_relevance}, we can first find that with the
	annotator relevance decreasing, the performance of some Type-I baselines (\eg, Forward and BLTM)
	also obviously decreases, since these methods model all noisy labels from one source. Second, as
	TAIDTM can effectively model the relevance between annotators via identifying and transferring
	their shared instance-dependent noise patterns, it consistently outperforms all baselines, which clearly
	demonstrates the effectiveness of the annotator- and instance-dependent modeling. Besides, the ablation study about the choice of mapping function, the choice of similarity measurement, and the sensitivity of hyperparameter $k$ can be found in Appendix~\ref{sensitivity}.

	\subsection{Visualization Results}
	\label{vis}
	As shown in the first two rows of Fig.~\ref{fig_v}, we visualize the process of building the similarity graph under different hyperparameters $k$ by reducing the dimension of the last layer’s parameters using the PCA algorithm, where each data point represents an annotator, and the data points with the same color means they have the same instance-dependent noise transition matrices. The edge between the two points indicates that they are adjacent. From the visualization, we can first find that all annotators with the same noise patterns will form a cluster in the proposed similarity space, showing that the similarity of the last layer’s parameters of different individual networks can measure the annotator similarity well. Second, we can see building the similarity graph is robust to $k$. When $k$ is less than 100, most of the neighboring annotators are similar annotators, and only when the $k$ is larger than 100, the graph will become very inaccurate since many dissimilar annotators will be regarded as nearest neighbors. In addition,
	by applying the Graph Purifying method~\cite{EntezariADP20}, the graph will be denoised, improving its robustness. Besides, 
	to further show the effectiveness of constructing the similarity graph, we also visualize the similarity graph under different numbers of annotator groups. As shown in the last two rows of Fig.~\ref{fig_v}, the proposed method can handle the similarity measure between annotators well under various annotator relevance. Note that to make the similarity graph more accurate when faced with large groups, we set $k = 30$ in the experiments under different numbers of annotator groups.
	
	\section{Conclusion}
	In this paper, we study a valuable problem of learning from crowds, \ie, estimating general annotator and instance-dependent transition matrices in practice. To address this problem when the annotations
	are sparse, we assume every annotator shares its noise patterns with similar annotators, and propose
	to estimate annotator- and instance-dependent transition matrices via knowledge transfer. Theoretical analyses justify the roles of both the knowledge transfer from global to individuals and the knowledge transfer between neighboring individuals.
	Empirical
	results on simulated and real-world crowd-sourcing datasets clearly verify the superiority of the
	proposed estimator. In the future, we are interested in extending our method to more scenes, such as robust image segmentation~\cite{yao2023learning} and data cleaning of foundation models~\cite{chowdhery2022palm}.
\section*{Acknowledgments}
This work was partially supported by grants from the Pioneer R\&D Program of Zhejiang Province (2024C01024), and Open Research Project of the State Key Laboratory of Media Convergence and Communication, Communication University of 
China (SKLMCC2022KF004). Tongliang Liu is partially supported by the following Australian Research Council projects: FT220100318, DP220102121, LP220100527, LP220200949, and IC190100031.
	
	\ifCLASSOPTIONcaptionsoff
	\fi

	
	
	
	\bibliographystyle{IEEEtran}
	\bibliography{IEEEabrv,final}
	
	%
	
	\begin{IEEEbiography}[{\includegraphics[width=1in,height=1.25in,clip,keepaspectratio]{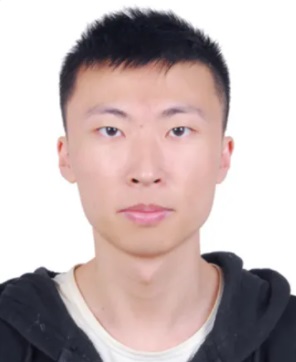}}]
		{Shikun Li} received the B.E. degree from the School of Information and Communication Engineering, Beijing University of Posts and Telecommunications (BUPT), Beijing, China. He is currently pursuing the Ph.D. degree with the Institute of Information Engineering, Chinese Academy of Sciences, Beijing, and the School of Cyber Security, University of Chinese Academy of Sciences, Beijing. His research interests include machine learning and computer vision.
	\end{IEEEbiography}
	\vspace{-30pt} 
	\begin{IEEEbiography}[{\includegraphics[width=1in,height=1.25in,clip,keepaspectratio]{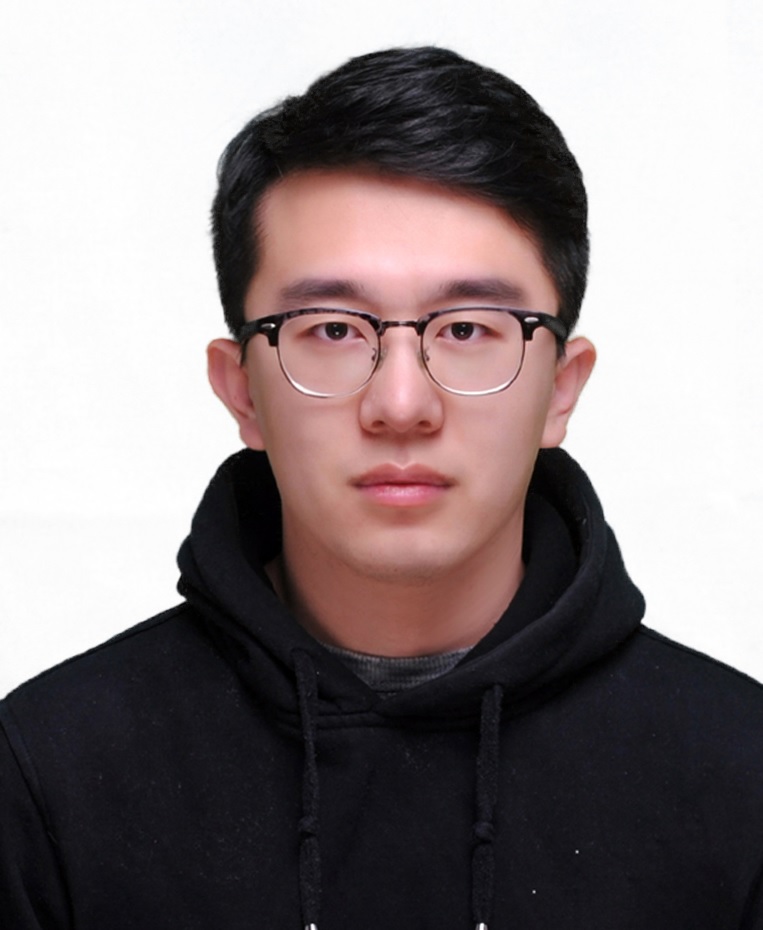}}]{Xiaobo Xia} received the B.E. degree in
		telecommunications engineering from Xidian University, in 2020. He is currently pursuing a Ph.D. degree in computer science from the University of Sydney. He has published
		more than 20 papers at top-tiered conferences
		or journals such as IEEE T-PAMI, ICML, ICLR, NeurIPS, CVPR, ICCV, and KDD. He also serves as the reviewer for top-tier conferences such as ICML, NeurIPS, ICLR, CVPR, ICCV, and ECCV. His research interest lies in machine learning, with a particular emphasis on weakly-supervised learning. He was a recipient of the Google Ph.D. Fellowship~(machine learning) in 2022. 
	\end{IEEEbiography}
	
	\vspace{-30pt} 
	\begin{IEEEbiography}[{\includegraphics[width=1in,height=1.25in,clip,keepaspectratio]{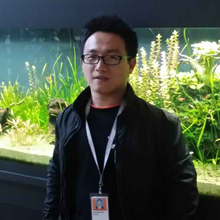}}]{Jiankang Deng} is currently an honorary lecturer at the Department of Computing, Imperial College London. He obtained his Ph.D. degree from Imperial College London in 2020. His research topic is deep learning-based face analysis and modeling. He has more than 12K+ citations to his research work. He is a reviewer in prestigious computer vision journals and conferences including T-PAMI, IJCV, CVPR, ICCV, and ECCV. He is one of the main contributors to the widely used open-source platform Insightface. He is a member of the IEEE.
	\end{IEEEbiography}
	
	\vspace{-30pt} 
	
\begin{IEEEbiography}[{\includegraphics[width=1in,height=1.25in,clip,keepaspectratio]{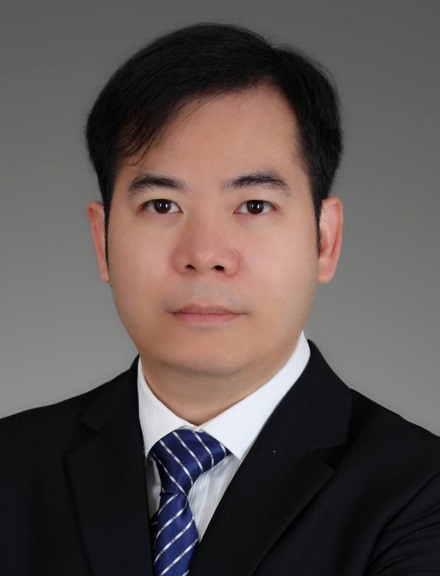}}]{Shiming Ge} (M'13-SM'15) is a Professor with the Institute of Information Engineering, Chinese Academy of Sciences. Prior to that, he was a senior researcher and project manager in Shanda Innovations, a researcher in Samsung Electronics and Nokia Research Center. He received the B.S. and Ph.D degrees both in Electronic Engineering from the University of Science and Technology of China (USTC) in 2003 and 2008, respectively. His research mainly focuses on computer vision, data analysis, machine learning and AI security, especially trustworthy learning solutions towards scalable applications. He is a senior member of IEEE, CSIG and CCF.
	\end{IEEEbiography}
	\vspace{-30pt} 
	\begin{IEEEbiography}[{\includegraphics[width=1in,height=1.25in,clip,keepaspectratio]{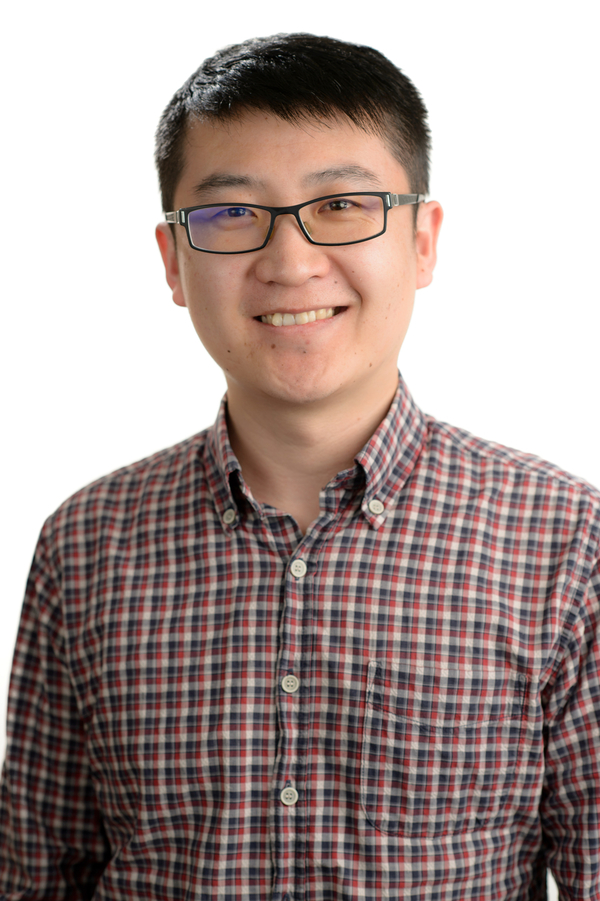}}]{Tongliang Liu} (Senior Member, IEEE) is an Associate Professor with the School of Computer Science and the Director of Sydney AI Centre at the University of Sydney. He is broadly interested in the fields of trustworthy machine learning and its interdisciplinary applications, with a particular emphasis on learning with noisy labels, adversarial learning, causal representation learning, transfer learning, unsupervised learning, and statistical deep learning theory. He has authored and co-authored more than 200 research articles including ICML, NeurIPS, ICLR, CVPR, ICCV, ECCV, AAAI, IJCAI, JMLR, and TPAMI. He is/was a (senior-) meta reviewer for many conferences, such as ICML, NeurIPS, ICLR, UAI, AAAI, IJCAI, and KDD, and was a notable AC for NeurIPS and ICLR. He is a co-Editor-in-Chief for Neural Networks, an Associate Editor of TMLR and ACM Computing Surveys, and is on the Editorial Boards of JMLR and MLJ. He is a recipient of CORE Award for Outstanding Research Contribution in 2024, the IEEE AI’s 10 to Watch Award in 2022, the Future Fellowship Award from Australian Research Council (ARC) in 2022, the Top-40 Early Achievers by The Australian in 2020, and the Discovery Early Career Researcher Award (DECRA) from ARC in 2018.
	\end{IEEEbiography}
	
	
	
	
	\onecolumn
	\begin{appendices}
		\section{Generation of Multiple Annotator- and Instance-dependent Noisy Labels}
		\label{gener_label_noise}
		As mentioned in Section~\ref{data}, to simulate instance-dependent noisy labels for multiple annotators, 
		following~\cite{yang2021estimating}, we corrupt training and validation sets manually with instance-dependent noise (IDN) according to the synthetic instance-dependent 
		transition matrices~\cite{xia2020part}. Algorithm~\ref{alg:algorithm2} lists the pseudo-code of the generation method. 
		In Section~\ref{data},
		we set $R=300, G=3, \bar{r}=2, \rho_{\max } = 0.6$
		with various individual noise rates to evaluate the 
		performance of all methods. In Section~\ref{data2}, we change the 
		average number of annotations for instances $\bar{r}$ and the number of annotator groups $G$ to study the influence of 
		annotation sparsity and annotator relevance, respectively.
		\begin{algorithm}[H]
			\renewcommand{\algorithmicrequire}{\textbf{Input:}}
			\renewcommand{\algorithmicensure}{\textbf{Output:}}
			\caption{Annotator- and Instance-dependent Label Noise Generation}
			\label{alg:algorithm2}
			\begin{algorithmic}[1]
				\REQUIRE Clean examples 
				$\{(\bm{x}_i, y_i)\}_{i=1}^n$, the
				number of annotators $R$, the number of annotator groups $G$, the average number of 
				annotations for instances $\bar{r}$,  individual noise rate $\rho$, individual noise rate bound $\rho_{max}$.
				\ENSURE Noisy examples $\left\{\left(\bm{x}_i, \left\{\bar{y}^j_i\right\}_{j\in w_i}\right)\right\}_{i=1}^n$.
				\FOR {$g=0,1,...,G-1$}
				\STATE Sample instance flip rates $q_i$ from the truncated normal distribution $\mathcal{N}\left(\rho, 0.1^2,\left[0, \rho_{\max }\right]\right)$;
				\STATE Independently sample $d_1, d_2, . . . , d_c$ from the standard normal distribution $\mathcal{N}\left(0,1^2\right)$;
				\FOR {$i=1,2,...,n$}
				\STATE $p=\bm{x}_i \times d_{y_i}$
				\quad {//generate instance-dependent flip rates}
				\STATE $p_{y_i}=-\infty$  \quad {//only consider entries that are different from the true label}
				\STATE  $p=q_i \times \operatorname{softmax}(p)$ {//make the sum of the off-diagonal entries of the $y_i$-th row to be $q_i$}
				\STATE $p_{y_i}=1-q_i$  \quad \quad {//set the diagonal entry to be $1-q_i$}
				\STATE  Randomly sample $\frac{R}{G}$ labels with replication from the label space according to the possibilities $p$ as noisy labels $\bar{y}^j_i$, where $j=g\frac{R}{G},g\frac{R}{G}+1, ..., (g+1)\frac{R}{G}-1$.
				\ENDFOR
				\ENDFOR
				\FOR {$i=1,2,...,n$}
				\STATE Randomly choose one annotator $j$ to let $j \in w_i$.  \quad  {//ensure each instance is labeled by at least one annotator}
				\ENDFOR
				\FOR {$j=1,2,...,R$}
				\STATE Randomly choose $(\frac{\bar{r}}{R}-1)n$ instances to let $j \in w_i$, where $i$ belongs to the indices of these instances.
				\ENDFOR
			\end{algorithmic}
		\end{algorithm}
		\section{Details of Baselines}
		\label{detail}
		In this paper, we exploit three types of baselines in experiments. Here, we detail these baselines. 
		
		Type-I baselines are the methods that treat all noisy labels from the same noisy distribution, which includes:
		\begin{itemize}
			\item CE uses the cross-entropy loss to train deep models from all noisy labels, without considering the side-effect of mislabeled data for generalization.
			\item GCE~\cite{zhang2018generalized} uses a generalized loss of CE and MAE to handle the label noise.
			\item Forward~\cite{PatriniRMNQ17} is a classifier-consistent method via loss correction with the estimated one-annotator transition matrices. 
			\item Reweight~\cite{Liu2016TPAMI} is a risk-consistent method via loss reweighting with the estimated one-annotator transition matrices. 
			\item BLTM~\cite{yang2021estimating} models the instance-dependent label noise via estimating one-annotator instance-dependent Bayes-label transition matrices.
		\end{itemize}
		
		Type-II baselines are model-free methods designed for learning from crowds, which do not explicitly model the noise generation process, which includes:
		\begin{itemize}
			\item DL-MV~\cite{IpeirotisPSW14} is training DNN on the result of majority voting, which assumes different annotators are equally reliable.
			\item DL-CRH~\cite{LiLGZFH14} is training DNN on the result of conflict resolution on the heterogeneous data model, which minimizes the overall weighted deviation between the truths and the multi-source observations. 
			\item DL-MMSR~\cite{MaO20} regards the problem of classification from crowdsourced data as the problem of reconstructing a rank-one matrix from a revealed subset of its corrupted entries, and proposes a new algorithm combining alternating minimization with extreme-value filtering to solve it.
			\item Max-MIG~\cite{CaoXKW19} is an information-theoretic method, which finds the information intersection between the data classifier and the label aggregator. The label aggregator is modeled by a weighted average function.
		\end{itemize} 
		
		Type-III baselines are model-based algorithms designed for learning from crowds, which model the noise generation process of each annotator by noise transition matrices, which include:
		\begin{itemize}
			\item DL-DS~\cite{Dawid1979Maximum} is
			training DNN on the result of the Dawid-Skene estimator, which estimates the transition matrices via the EM algorithm.
			\item DL-IBCC~\cite{KimG12} is training DNN on the result of the independent Bayesian classifier combination, which estimates the transition matrices with a Bayesian model.
			\item DL-EBCC~\cite{LiRC19} is training DNN on the result of the enhanced Bayesian classifier combination, which captures worker correlation by modeling true classes as mixtures of subtypes.
			\item AggNet~\cite{Albarqouni2016TMI} uses the EM algorithm to jointly estimate the transition matrices and data classifier.
			\item CrowdLayer~\cite{Rodrigues2018aaai}  adds a crowd layer to the output of a common network to learn the transition matrices in an end-to-end fashion.
			\item MBEM~\cite{Khetan2018iclr} is an improved EM algorithm that rewrites the EM likelihood of AggNet and regards the estimated true labels as hard labels.
			\item UnionNet~\cite{Wei2022DeepLF} takes all labels from different annotators as a union, and maximizes the likelihood of this union with only a parametric transition matrix.
			\item CoNAL~\cite{Chu0W21} decomposes annotation noise into common noise and individual noise and differentiates the source of confusion based on instance difficulty and annotator expertise. It can be regarded as estimating a special type of AIDTM, where one noise pattern is shared among all annotators.
		\end{itemize}
\section{Discussion about the choice of mapping function}
		\label{f2}
		As discussed in Section~\ref{theory2}, there are two roles of the mapping function. If one function
		can play such roles well, it can also be used as the mapping function. Hence, all functions that can capture the similarity of annotators from graph and keep this similarity in their output embeddings, are alternatives. Here, we test one more
		complex architecture (GAT~\cite{VelickovicCCRLB18}), and one simpler method (Deepwalk~\cite{perozzi2014deepwalk}) to replace the GCN-based function. GAT~\cite{VelickovicCCRLB18} is a GNN architecture, which introduces masked self-attentional layers
		with multi-head attention into GCN to capture graph structure. For Deepwalk~\cite{perozzi2014deepwalk} method, we
		first perform Random Walk on the similarity graph to generate node representations and then use
		a two-layer non-linear transformation to map node representations into the last layer’s parameters. As
		shown in Table~\ref{mapf}, both TAIDTM w/ GAT and TAIDTM w/ Deepwalk lead to a worse performance
		than TAIDTM w/ GCN, which may be because GAT introduces too many parameters to learn from sparse annotations and
		Deepwalk learns in a two-stage way, which cannot directly learn to merge node representations.
		
		\begin{table}[h]
			\caption{The performance ($\%$) with different mapping functions. }
			\centering
			\fontsize{10pt}{12pt}\selectfont
			\setlength\tabcolsep{5pt}
			\begin{tabular}{l|ccc|c}
				\hline
				Datasets & LabelMe & Music & CIFAR10-N & Avg.\\
				\hline  
				TAIDTM w/ GAT~\cite{VelickovicCCRLB18} & 86.16$\pm$0.41 & 68.56$\pm$1.62 & 83.84$\pm$0.15 & 79.52\\
				TAIDTM w/ Deepwalk~\cite{perozzi2014deepwalk}  & 85.36$\pm$0.78 & 67.23$\pm$1.33 & 83.90$\pm$0.23 & 78.83\\	
				TAIDTM w/ GCN  & \textbf{86.68$\pm$0.49} & \textbf{69.39$\pm$1.96} & \textbf{83.99$\pm$0.51} & \textbf{80.02}\\	
				\hline
			\end{tabular}
			\label{mapf}
		\end{table}
		
		\section{Discussion about the choice of similarity measurement}
		In this section, we try to use the Jaccard distance between label vectors~\cite{kairam2016parting}, the similarity of class confusions~\cite{Chu0W21}, and the community results of CBCC~\cite{VenanziGKKS14} to measure annotator similarity and construct the similarity graph. The experiments of TAIDTM with these similarity estimation ways on real-world datasets are conducted. As shown in Table~\ref{simil}, due to the effectiveness of the constructed graph based on instance-dependent noise patterns, our TAIDTM performs better than others consistently.
		\begin{table}[h]
			\caption{The performance ($\%$) with different similarity measurements. }
			\centering
			\fontsize{10pt}{12pt}\selectfont
			\setlength\tabcolsep{5pt}
			\begin{tabular}{l|ccc|c}
				\hline
				Datasets & LabelMe & Music & CIFAR10-N & Avg.\\
				\hline  
				TAIDTM w/~\cite{kairam2016parting} & 86.47$\pm$0.54 & 66.85$\pm$1.27 & 83.46$\pm$0.34 & 78.93\\
				TAIDTM w/~\cite{Chu0W21}  & 86.44$\pm$0.58 & 66.77$\pm$1.53 & 83.38$\pm$0.59 & 78.86\\	
				TAIDTM w/~\cite{VenanziGKKS14}  & 86.59$\pm$0.53& 66.81$\pm$1.17 & 83.37$\pm$0.19 & 78.92\\	
				\hline
				TAIDTM  & \textbf{86.68$\pm$0.49} & \textbf{69.39$\pm$1.96} & \textbf{83.99$\pm$0.51} & \textbf{80.02}\\	
				\hline
			\end{tabular}
			\label{simil}
		\end{table}
		
		\section{Sensitivity of hyperparameter $k$}\label{sensitivity}
		To study the sensitivity of the number of nearest neighbors $k$, we run the TAIDTM method on CIFAR10
		dataset with AIDN-40$\%$ label noise under different choices of $k$. According to Section~\ref{data}, there are
		100 noisy annotators in each group. Therefore, if the similarities between annotators are accurate,
		the best $k$ is around 100. As shown in Table~\ref{senstive}, we can find that our AIDTM is very robust
		to the choice of $k$ when $k$ is less than 100, and only when the $k$ is larger than 100, the performance
		of TAIDTM will decline obviously. Besides, as reported in Table~\ref{senstive}, applying the Graph Purifying method~\cite{EntezariADP20} in constructing the similarity graph can further improve the robustness, leading to better
		performance in most cases.
		
		\begin{table}[h]
			\caption{The performance ($\%$) with different hyperparameters $k$. }
			\fontsize{10pt}{12pt}\selectfont
			\setlength\tabcolsep{5pt}
			\centering
			\begin{tabular}{l|cccccc}
				\hline
				Hyperparameter $k$  & 1 &10 & 30 & 50 & 100 & 200\\
				\hline  
				TAIDTM w/o Graph Purifying & 79.90 & 80.16 & 79.81 & 79.88 & 80.19  & 75.33\\
				TAIDTM w/ Graph Purifying  & 80.42 & 79.70 & 80.01  &79.95 & 80.28 & 78.77\\	
				\hline
			\end{tabular}
			\label{senstive}
		\end{table}	

\section{Experiments on a larger dataset}
{\color{black}
Following the previous work~\cite{gao2022learning}, we conducted an experiment on a larger dataset, ImageNet100, which is a commonly used subset of ImageNet dataset~\cite{tian2020contrastive}, containing 100 random classes. Following the simulated method described in the paper, we create 300 annotators with AIDN-40$\%$ label noise. We choose ResNet18 as the backbone model. For all methods, we train deep classifiers for 20 epochs using SGD optimizer with an initial learning rate of 0.2 and weight decay of 1e-4, and the learning rate is divided by 3 after 10 and 15 epochs. 

Due to the limited computing resources, we mainly compared our TAIDTM with the state-of-the-art methods designed for deep learning from crowds. As shown in Tab.~\ref{imagenet}, our TAIDTM outperforms other baselines on the ImageNet100 dataset with AIDN-40$\%$ label noise.}
\begin{table}[h]
\color{black}
\caption{\color{black} Comparison with state-of-the-art methods in the test accuracy (\%) on simulated ImageNet100 dataset. The best results are in \best{bold}.}
\centering
\label{imagenet}
			\fontsize{10pt}{12pt}\selectfont
\setlength\tabcolsep{5pt}
	\begin{tabular}{c|c}
	\hline Method & ImageNet100-AIDN-40\% \\
	\hline
	DL-MV &  39.09$\pm$0.89   \\
	DL-DS &  41.39$\pm$1.71  \\
	Max-Mig &  40.45$\pm$0.35  \\
	AggNet &  50.55$\pm$1.37   \\
	CrowdLayer &  9.22$\pm$1.10  \\
	MBEM & 9.80$\pm$0.70   \\
	UnionNet &  12.64$\pm$0.78   \\
	CoNAL & 4.51$\pm$0.90   \\
	\hline TAIDTM &  \textbf{51.97$\pm$1.11}  \\
	\hline
\end{tabular}
\end{table}	
		\section{Proof of Theorem~\ref{thm1}}
		\label{proof_1}
		We start by introducing the Rademacher complexity method~\cite{bartlett2003rademacher} for deriving generalization bounds.
		
		\begin{Definition}[\cite{bartlett2003rademacher}]
			Let $\sigma_1,\ldots,\sigma_n$ be independent Rademacher variables, $X_1,\ldots,X_n$ be i.i.d. variables and $F$ be a real-valued hypothesis space. The Rademacher complexity of the hypothesis space over the variable is defined as 
			\[\mathfrak{R}(F)=\mathbb{E}_{X,\sigma}\left[\sup_{f\in F}\frac{1}{n}\sum_{i=1}^{n}\sigma_if(X_i)\right].\]
		\end{Definition}
		
		The following theorem plays an important role in deriving the generalization bounds.
		\begin{theorem}[\cite{bartlett2003rademacher}]\label{generR}
			Let $F$ be an $[a,b]$-valued hypothesis space on $\mathcal{X}$, and $S=\{X_1,\ldots,X_n\}\in\mathcal{X}^n$. Then, for any $f\in F$ and any $\delta\in (0,1)$, with probability at least $1-\delta$, we have
			\begin{eqnarray*}
				\mathbb{E}_X[f(X)]-\frac{1}{n}\sum_{i=1}^{n}f(X_i)\leq 2\mathfrak{R}(F)+(b-a)\sqrt{\frac{\log(1/\delta)}{2n}}.
			\end{eqnarray*}
		\end{theorem}

	According to Theorem \ref{generR}, since we assume the cross-entropy loss function is upper bounded by $M$, it is easy to get that for any $j \in [R]$, $\delta\in (0,1)$, with probability at least $1-\delta$,
	\begin{equation}
		\label{proof1}
		{\mathcal{R}}_{D_I^j}(\hat{T}^j)-\hat{\mathcal{R}}_{\mathcal{D}_I^j}(\hat{T}^j) \leq 2\mathfrak{R}(\ell \circ \mathcal{T})+M\sqrt{\frac{\log({{1}/{\delta}})}{2m_j}}.
	\end{equation}

    Recall that the $(i, j)$-th entry of the transition matrix $T$ is obtained by $T_{ij}(\boldsymbol{x})=\exp \left(t_{ij}(\boldsymbol{x})\right) / \sum_{k=1}^C \exp \left(t_{ik}(\boldsymbol{x})\right)$, and $t(\boldsymbol{x})$ is defined by a $d$-layer neural network, \ie, $t: \boldsymbol{x}\mapsto {W}_d\sigma_{d-1}({W}_{d-1}\sigma_{d-2}(\ldots \sigma_1({W}_1\boldsymbol{x})))\in \mathbb{R}^{C\times C}$,  $W_1,\ldots,W_d$ are the parameter matrices, and $\sigma_1,\ldots,\sigma_{d-1}$ are activation functions. To further upper bound the Rademacher complexity, we need to consider the Lipschitz continuous property of the loss function \textit{w.r.t.} $t_{ij}(\boldsymbol{x})$. 
    
    We can further upper bound the Rademacher complexity $\mathfrak{R}(\ell \circ \mathcal{T})$ by the following lemma.
	  
	\begin{Lemma} \label{lemma1}
 Let $n$ be the size of training examples. Let $\mathcal{D}_k$ be the instance set of  training examples whose  ${y}^{\star}=k, k\in [C]$. Then, we have
		\begin{align*}
			\mathfrak{R}(\ell \circ \mathcal{T})=\mathbb{E}\left[\sup_{f}\frac{1}{n}\sum_{i=1}^{n}\sigma_i\ell(\bar{\boldsymbol{y}}_i, \boldsymbol{y}^{\star}_i \cdot T(\boldsymbol{x}_i))\right]\leq  C \sum_{k=1}^C  \frac{|\mathcal{D}_k|}{n}\mathbb{E}\left[ \sup_{ t \in H} \frac{1}{|\mathcal{D}_k|}\sum_{\boldsymbol{x}'_i \in \mathcal{D}_k}\sigma_i t(\boldsymbol{x}'_i)\right],
		\end{align*}
		where $H$ is the hypothesis space induced by the deep neural network.
	\end{Lemma}
	A detailed proof of Lemma \ref{lemma1} is provided in Appendix \ref{sec:app3}.
	
	Note that $\mathbb{E}\left[ \sup_{t\in H }\frac{1}{|\mathcal{D}_k|}\sum_{\boldsymbol{x}'_i \in \mathcal{D}_k}\sigma_i t(\boldsymbol{x}'_i)\right]$ measures the hypothesis complexity of deep neural networks, which can be bounded by the following theorem.
	\begin{theorem}[\cite{golowich2018size}]\label{thm:network}
		Assume the Frobenius norm of the weight matrices $W_1,\ldots,W_d$ are at most $M_1,\ldots, M_d$. Let the activation functions be 1-Lipschitz, positive-homogeneous, and applied element-wise (such as the ReLU). Let $n$ be the size of training examples. Let $\boldsymbol{x}$ is upper bounded by B, i.e., for any $\boldsymbol{x}
		\in \mathcal{X}$, $\|\boldsymbol{x}\|\leq B$. Then,
		\begin{align*}
			\mathbb{E}\left[ \sup_{t\in H }\frac{1}{n}\sum_{i=1}^{n}\sigma_it(\boldsymbol{x}_i)\right]\leq \frac{B(\sqrt{2d\log2}+1)\Pi_{i=1}^{d}M_i}{\sqrt{n}}.
		\end{align*}
	\end{theorem}

    According to Lemma \ref{lemma1}, and Theorem~\ref{thm:network},  let $\mathcal{D}_k^j$ be the instance set of  distilled training examples whose  ${y}^{\star}=k$ from the annotator $j$, and we have
		\begin{align}
			\mathfrak{R}(\ell \circ \mathcal{T})
    & \leq C \sum_{k=1}^C\mathbb{E}\left[ \sup_{t\in H }\frac{1}{|\mathcal{D}_k^j|}\sum_{\boldsymbol{x}'_i \in \mathcal{D}_k^j}\sigma_i t(\boldsymbol{x}'_i)\right] \nonumber \\
    & \leq C \sum_{k=1}^C \frac{|\mathcal{D}_k^j|}{m_j} \frac{B(\sqrt{2d\log2}+1)\Pi_{i=1}^{d}M_i}{\sqrt{|\mathcal{D}_k^j|}} \nonumber \\
    & = \frac{{C}^2B(\sqrt{2d\log2}+1)\Pi_{i=1}^{d}M_i}{\sqrt{Cm_j}},
    \label{proof7}
		\end{align}    
    where the last equation holds because we assume the distilled training set $\mathcal{D}_I^j$ is class-balanced.
    
    According to Eq.~(\ref{proof1}) and Eq.~(\ref{proof7}), the proof of Theorem~\ref{thm1} is concluded.
	
	\subsection{Proof of Lemma~\ref{lemma1}}
	\label{sec:app3}
 
Before proving Lemma~\ref{lemma1}, we show that the loss function $\ell(\bar{\boldsymbol{y}}, \boldsymbol{y}^{\star} \cdot T(\boldsymbol{x}))$ is 1-Lipschitz-continuous \textit{w.r.t.}  $t_{ij}(\boldsymbol{x}),i,j\in\{1,\ldots,C\}$. 
	
Note that
\begin{equation}
	\ell(\bar{\boldsymbol{y}}, \boldsymbol{y}^{\star} \cdot T(\boldsymbol{x})) = -\log(\bar{\boldsymbol{y}}^\top \boldsymbol{y}^{\star} T(\boldsymbol{x})) = -\log(T_{{y}^{\star}\bar{y}}(\boldsymbol{x}))  = -\log\left(\frac{\exp(t_{{y}^{\star}\bar{y}}(\boldsymbol{x}))}{\sum_{i=1}^C \exp(t_{{y}^{\star}i}(\boldsymbol{x})))}\right).
\end{equation}

Take the derivative of $\ell(\bar{\boldsymbol{y}}, \boldsymbol{y}^{\star} \cdot T(\boldsymbol{x}))$ \textit{w.r.t.} $t_{ij}(\boldsymbol{x})$. If $i \neq {y}^{\star}$, we have 
\begin{equation} \label{derivative1}
\begin{aligned}
&\frac{\partial {\ell}(\bar{\boldsymbol{y}}, \boldsymbol{y}^{\star} \cdot T(\boldsymbol{x}))}{\partial t_{ij}(\boldsymbol{x})} = 0.
\end{aligned}
\end{equation}
If $i = {y}^{\star}$ and $j \neq \bar{y}$, we have 
\begin{equation} \label{derivative2}
\begin{aligned}
&\frac{\partial {\ell}(\bar{\boldsymbol{y}}, \boldsymbol{y}^{\star} \cdot T(\boldsymbol{x}))}{\partial t_{ij}(\boldsymbol{x})} = \frac{\exp(t_{ij}(\boldsymbol{x}))}{\sum_{k=1}^C \exp(t_{ik}(\boldsymbol{x})))}.
\end{aligned}
\end{equation}
If $i = {y}^{\star}$ and $j = \bar{y}$, we have 
\begin{equation} \label{derivative3}
\begin{aligned}
&\frac{\partial {\ell}(\bar{\boldsymbol{y}}, \boldsymbol{y}^{\star} \cdot T(\boldsymbol{x}))}{\partial t_{ij}(\boldsymbol{x})} = -1+\frac{\exp(t_{ij}(\boldsymbol{x}))}{\sum_{k=1}^C \exp(t_{ik}(\boldsymbol{x})))}.
\end{aligned}
\end{equation}

According to Eqs. (\ref{derivative1}), (\ref{derivative2}), and (\ref{derivative3}), it is easy to conclude that $-1 \leq \frac{\partial {\ell}(\bar{\boldsymbol{y}}, \boldsymbol{y}^{\star} \cdot T(\boldsymbol{x}))}{\partial t_{ij}(\boldsymbol{x})} \leq 1$, which also indicates that the loss function is 1-Lipschitz with respect to $t_{ij}(\boldsymbol{x}),i,j\in\{1,\ldots,C\}$. 

Then, we have
\begin{equation}
	\label{proof4}
	\begin{aligned}
		&\mathbb{E}\left[\sup_{T}\frac{1}{n}\sum_{i=1}^{n}\sigma_i\ell(\bar{\boldsymbol{y}}_i, \boldsymbol{y}^{\star}_i \cdot T(\boldsymbol{x}_i))\right]\\
		&=\mathbb{E}\left[\sup_{T}\frac{1}{n}\sum_{i=1}^{n}\sum_{k=1}^C \mb{I}[{y}^{\star}_i=k]\sigma_i\ell(\bar{\boldsymbol{y}}_i,T_k(\boldsymbol{x}_i))\right]\\
		&=\mathbb{E}\left[\sup_{T}\sum_{k=1}^C \frac{1}{n}\sum_{i=1}^{n} \mb{I}[{y}^{\star}_i=k]\sigma_i\ell(\bar{\boldsymbol{y}}_i,T_k(\boldsymbol{x}_i))\right]\\
		&\leq \mathbb{E}\left[\sum_{k=1}^C\sup_{T_k} \frac{1}{n}\sum_{i=1}^{n} \mb{I}[{y}^{\star}_i=k]\sigma_i\ell(\bar{\boldsymbol{y}}_i,T_k(\boldsymbol{x}_i))\right]\\
	&=	\sum_{k=1}^C \mathbb{E}\left[ \sup_{T_k} \frac{1}{n}\sum_{i=1}^{n} \mb{I}[{y}^{\star}_i=k]\sigma_i\ell(\bar{\boldsymbol{y}}_i,T_k(\boldsymbol{x}_i))\right]\\
	&=	\sum_{k=1}^C \mathbb{E}\left[ \sup_{T_k} \frac{1}{n}\sum_{ \boldsymbol{x}'_i\in \mathcal{D}_k} \sigma_i\ell(\bar{\boldsymbol{y}},T_k(\boldsymbol{x}'_i))\right],
	\end{aligned}
\end{equation}
where $\mb{I}[\cdot]$ is the indicator function which takes 1 if the identity index is true and 0 otherwise.

Furthermore, for $\mathcal{D}_k$,  we can get
\begin{equation}
		\label{proof5}
	\begin{aligned}
		&\mathbb{E}\left[ \sup_{T_k} \frac{1}{n}\sum_{ \boldsymbol{x}'_i\in \mathcal{D}_k} \sigma_i\ell(\bar{\boldsymbol{y}}',T_k(\boldsymbol{x}'_i))\right]\\
		&=\frac{|\mathcal{D}_k|}{n}\mathbb{E}\left[ \sup_{T_{k}} \frac{1}{|\mathcal{D}_k|}\sum_{\boldsymbol{x}'_i \in \mathcal{D}_k}\sigma_i\ell(\bar{\boldsymbol{y}}',T_k(\boldsymbol{x}'_i))\right] \\
		&=\frac{|\mathcal{D}_k|}{n}\mathbb{E}\left[ \sup_{ \max(t_{k1},...,t_{kC})} \frac{1}{|\mathcal{D}_k|}\sum_{\boldsymbol{x}'_i \in \mathcal{D}_k}\sigma_i\ell(\bar{\boldsymbol{y}}'_i,T_k(\boldsymbol{x}'_i))\right] \\
		& \leq \frac{|\mathcal{D}_k|}{n}\mathbb{E}\left[ \sum_{j=1}^C \sup_{ t_{kj} \in H} \frac{1}{|\mathcal{D}_k|}\sum_{\boldsymbol{x}'_i \in \mathcal{D}_k}\sigma_i\ell(\bar{\boldsymbol{y}}'_i,T_k(\boldsymbol{x}'_i))\right] \\
		& = \sum_{j=1}^C\frac{|\mathcal{D}_k|}{n}\mathbb{E}\left[ \sup_{ t_{kj} \in H} \frac{1}{|\mathcal{D}_k|}\sum_{\boldsymbol{x}'_i \in \mathcal{D}_k}\sigma_i\ell(\bar{\boldsymbol{y}}'_i,T_k(\boldsymbol{x}'_i))\right] \\
  & \leq C \frac{|\mathcal{D}_k|}{n}\mathbb{E}\left[ \sup_{ t_{kj} \in H} \frac{1}{|\mathcal{D}_k|}\sum_{\boldsymbol{x}'_i \in \mathcal{D}_k}\sigma_i t_{kj}(\boldsymbol{x}'_i)\right]\\
  & = C \frac{|\mathcal{D}_k|}{n}\mathbb{E}\left[ \sup_{ t \in H} \frac{1}{|\mathcal{D}_k|}\sum_{\boldsymbol{x}'_i \in \mathcal{D}_k}\sigma_i t(\boldsymbol{x}'_i)\right],
			\end{aligned}
	\end{equation}
where the fifth inequality holds because of the Talagrand Contraction Lemma~\cite{ledoux1991probability} and the proved property that the loss function is 1-Lipschitz-continuous \textit{w.r.t.}  $t_{kj}(\boldsymbol{x}'_i)$.

According to Eq.~(\ref{proof4}) and Eq.~(\ref{proof5}), we have
\begin{equation}
		\label{proof6}
\mathbb{E}\left[\sup_{T}\frac{1}{n}\sum_{i=1}^{n}\sigma_i\ell(\bar{\boldsymbol{y}}_i, \boldsymbol{y}^{\star}_i \cdot T(\boldsymbol{x}_i))\right] \leq C \sum_{k=1}^C  \frac{|\mathcal{D}_k|}{n}\mathbb{E}\left[ \sup_{ t \in H} \frac{1}{|\mathcal{D}_k|}\sum_{\boldsymbol{x}'_i \in \mathcal{D}_k}\sigma_i t(\boldsymbol{x}'_i)\right].
\end{equation}

%
%
		
		\section{Proof of Theorem~\ref{thm2}}
		\label{proof_2}
		We start by introducing the generalization loss of a stochastic hypothesis and its corresponding empirical loss.
		\begin{Definition}[\cite{shalev2014understanding,mcnamara2017risk}]
	Let $D$ be an arbitrary distribution over an example domain $Z$. Let $\mathcal{H}$ be a hypothesis space and let $\ell: \mathcal{H} \times Z \rightarrow[0,M]$ be the loss function. Let $\mathcal{D} = \{z_1, ...,z_m\}$ be an i.i.d. training set sampled according to $D$. The generalization loss of $Q$ is defined as 
	$$\mathcal{R}'_{{D}}(Q):=\mathbb{E}_{z \sim {D}, h \sim Q}[\ell(h,z)],$$ and the corresponding empirical loss is defined as
	$$\hat{\mathcal{R}}'_{{\mathcal{D}}}(Q):=\mathbb{E}_{h \sim Q}[\frac{1}{m}\sum_{i=1}^m\ell(h,z_m)].$$
\end{Definition}		
		
			\begin{Lemma}[\cite{shalev2014understanding}] \label{lemma2}
			Let $D$ be an arbitrary distribution over an example domain $Z$. Let $\mathcal{H}$ be a hypothesis space and let $\ell: \mathcal{H} \times Z \rightarrow[0,1]$ be the loss function. Let $\mathcal{D} = \{z_1, ...,z_m\}$ be an i.i.d. training set sampled according to $D$. Assuming $P$ be a prior distribution over $\mathcal{H}$. Then, for any distribution $Q$ over $\mathcal{H}$, $\delta\in (0,1)$, with probability at least $1-\delta$,
			$$
			{\mathcal{R}}'_{{D}}(Q) \leq \hat{\mathcal{R}}'_\mathcal{D}(Q)+\sqrt{\frac{KL(Q \| P)+ \log (m / \delta)}{2(m-1)}}.
			$$
		\end{Lemma}
		
		According to Lemma~\ref{lemma2}, by making the losses normalized, with probability at least 1 - $\delta/2$, we have 
		\begin{align}
			\label{proof2}
	\frac{{\mathcal{R}}'_{D_I^j}(\tilde{T}^j)}{M}-\frac{\hat{\mathcal{R}}'_{\mathcal{D}_I^j}(\tilde{T}^j)}{M} & 
	\leq \sqrt{\frac{ K L\left(\tilde{T}^j \| \tilde{T}_{G}\right)+\log({2 m_j}/{\delta})}{2\left(m_j-1\right)}} \nonumber \\
	& \leq \sqrt{\frac{ \omega\left(\mathcal{R}_{D_G}(\hat{T})\right)+\log({2 m_j}/{\delta})}{2\left(m_j-1\right)}},
\end{align}
where the second inequality holds because of the assumption that $\forall \tilde{T} \in \tilde{\mathcal{T}}, K L\left(\tilde{T} \| \tilde{T}_{G}\right) \leq \omega\left(\mathcal{R}_{D_G}(\hat{T})\right)$.

Furthermore, using the same tricks as the proof of Theorem~\ref{thm1}, it is easy to prove the following theorem.
\begin{theorem} Assume the Frobenius norm of the weight matrices ${W}_1,\ldots,{W}_d$ are at most $M_1,\ldots, M_d$, and the instances $\boldsymbol{x}$ are upper bounded by $B$, \ie, $\|\boldsymbol{x}\|\leq B$ for all $\boldsymbol{x}\in\mathcal{{X}}$. Let the activation functions be 1-Lipschitz, positive-homogeneous, and applied element-wise (such as the ReLU). Then, for any $j \in [R]$, $\delta\in (0,1)$, with probability at least $1-\delta$,
	\label{thm3}
	\begin{align}
		{\mathcal{R}}_{D_G}(\hat{T})-\hat{\mathcal{R}}_{\mathcal{D}_G}(\hat{T}) \leq \frac{2BC(\sqrt{2d\log2}+1)\Pi_{i=1}^{d}M_i}{\sqrt{Cmr}}+M\sqrt{\frac{\log({{1}/{\delta}})}{2mr}}.\nonumber
	\end{align}
\end{theorem}

Hence, since we assume $\hat{\mathcal{R}}_{\mathcal{D}_G}(\hat{T})\approx 0$,  then with probability at least $1-\delta/2$,
	\begin{align}
		\label{proof3}
	{\mathcal{R}}_{D_G}(\hat{T}) \leq \frac{2C^2B(\sqrt{2d\log2}+1)\Pi_{i=1}^{d}M_i}{\sqrt{Cmr}}+M\sqrt{\frac{\log({{2}/{\delta}})}{2mr}} = O(\frac{1}{\sqrt{mr}}).
\end{align}

By combining Eq.~(\ref{proof2}) and Eq.~(\ref{proof3}), the proof of Theorem~\ref{thm2} is completed.

\section{Illustrative example for the role of knowledge transfer between Neighboring Individuals}
\label{example}	
{\color{black} Here, we provide an illustrative example to make how the knowledge transfer between neighboring individuals can work on similar and different annotators more clear. Assuming we use one GCN layer, according to Section 3.4, the GCN-based mapping function will be $\mathbf{H}^1=h\left(\widehat{\mathbf{A}}^* \mathbf{H}^0 \mathbf{W}\right)$, where
	$\widehat{\mathbf{A}}^*$ is the normalized adjacency matrix, $\mathbf{H}^0$ is the input node features, $\mathbf{W}$ is a transformation
	matrix to be learned in this layer, and $h(\cdot)$ denotes a non-linear operation. If there are a total of $R = 4$ annotators, input node features are $H^0=\left(\begin{array}{llll}1 & 0 & 0 & 0 \\ 0 & 1 & 0 & 0 \\ 0 & 0 & 1 & 0 \\ 0 & 0 & 0 & 1\end{array}\right)$ where each row vector represents one annotator. Without loss of generality, we suppose these
	4 annotators can be divided into two groups of 2 annotators, and one group has the same noise pattern, then an accurate adjacency matrix is $\widehat{\mathbf{A}}^*=\left(\begin{array}{cccc}0.5 & 0.5 & 0 & 0 \\ 0.5 & 0.5 & 0 & 0 \\ 0 & 0 & 0.5 & 0.5 \\ 0 & 0 & 0.5 & 0.5\end{array}\right)$. According
	to the above mapping function, the node features of neighboring annotators will first be merged as $\widehat{\mathbf{A}}^* \mathbf{H}^0=\widehat{\mathbf{A}}^*=\left(\begin{array}{cccc}0.5 & 0.5 & 0 & 0 \\ 0.5 & 0.5 & 0 & 0 \\ 0 & 0 & 0.5 & 0.5 \\ 0 & 0 & 0.5 & 0.5\end{array}\right)$. After that, $\widehat{\mathbf{A}}^* \mathbf{H}^0$ will be transformed into output
	node features $\mathbf{H}^1$ by $\mathbf{W}$ and $h(\cdot)$, \ie,
	\begin{small}
		$$
		\mathbf{H}^1=h\left(\widehat{\mathbf{A}}^* \mathrm{H}^0 \mathbf{W}\right)=\left(\begin{array}{llll}
			h\left(0.5 *\left(W_{00}+W_{10}\right)\right) & h\left(0.5 *\left(W_{01}+W_{11}\right)\right) & h\left(0.5 *\left(W_{02}+W_{12}\right)\right) & h\left(0.5 *\left(W_{03}+W_{13}\right)\right) \\
			h\left(0.5 *\left(W_{00}+W_{10}\right)\right) & h\left(0.5 *\left(W_{01}+W_{11}\right)\right) & h\left(0.5 *\left(W_{02}+W_{12}\right)\right) & h\left(0.5 *\left(W_{03}+W_{13}\right)\right) \\
			h\left(0.5 *\left(W_{20}+W_{30}\right)\right) & h\left(0.5 *\left(W_{21}+W_{31}\right)\right) & h\left(0.5 *\left(W_{22}+W_{32}\right)\right) & h\left(0.5 *\left(W_{23}+W_{33}\right)\right) \\
			h\left(0.5 *\left(W_{20}+W_{30}\right)\right) & h\left(0.5 *\left(W_{21}+W_{31}\right)\right) & h\left(0.5 *\left(W_{22}+W_{32}\right)\right) & h\left(0.5 *\left(W_{23}+W_{33}\right)\right)
		\end{array}\right).
		$$
	\end{small}
	We can first find that after employing a GCN-based mapping function, the annotators of one group obtain the same final node features, which will be used as the last layer’s parameters of individual noise-transition networks, leading to the same estimated AIDTM. 
	In addition, the labeled data from the annotators of one group will be used together to learn the same part of parameters $\mathbf{W}$, achieving the knowledge transfer between them and alleviating the modeling challenge caused by sparse annotations.}

	\end{appendices}
\end{document}